  \documentclass[reqno]{amsart}

\usepackage{amsmath, amssymb}
 \usepackage{changebar}
\usepackage{color}
\usepackage{enumerate, xspace}
\numberwithin{equation}{section}
\newtheorem{thm}{Theorem}[section]
\newtheorem{lem}[thm]{Lemma}
\newtheorem{cor}[thm]{Corollary}

\newtheorem{prop}[thm]{Proposition}

\newtheorem{defin}[thm]{Definition}

\newtheorem{rem}[thm]{Remark}
  
\renewcommand\Pr{{\mathbb P}}
\newcommand\N{{\mathbb N}}

\newcommand\E{{\mathbb E}}
\newcommand\R{{\mathbb R}}
\newcommand\Z{{\mathbb Z}}
\newcommand\1{{1\kern-.25em\hbox{\rm I}}}
\newcommand\eu{{1\kern-.25em\hbox{\sm I}}}

\newcommand\BB{{\mathcal B}}

\newcommand\D{{\mathcal D}}

\newcommand\FF{{\mathcal F}}

\newcommand\KK{{\mathcal K}}

\newcommand\WW{{\mathcal W}}

\newcommand\RR{{\mathcal R}}

 \newcommand\e{\epsilon}
\newcommand\eps{\epsilon}

\newcommand\s{\sigma}

\newcommand\om{\omega}
 \newcommand\La{\Lambda}
\newcommand\G{\Gamma}
\newcommand\Om{\Omega}

\newcommand{\dist}{{\rm dist}}

\newcommand{\nada}[1]{}

\begin{document}
 \def\nic#1{\textcolor{red}{#1}}
   \def\en#1{\textcolor{blue}{#1}}
 \title[Unique minimizer for a Random functional]
{Uniqueness of the
 minimizer for a random nonlocal  functional  with double-well potential in $d\le2$.   }
\author{Nicolas Dirr}\thanks{N.D. supported by GNFM-INDAM. }
\address{Nicolas Dirr, Cardiff School of Mathematics
Cardiff University
Senghennydd Road, Cardiff, Wales, UK, CF24 4AG.}
\email{ {\tt DirrNP@cardiff.ac.uk}}
\author{Enza Orlandi}\thanks{E.O. supported by MURST/Cofin  
Prin-2009-2013  
and ROMA TRE University}
\address{Enza Orlandi, 
Dipartimento di Matematica\\
Universit\`a  di Roma Tre\\
 L.go S.Murialdo 1, 00146 Roma, Italy. }
\email{{\tt orlandi@mat.uniroma3.it}}
\date{\today}
\begin{abstract} We  consider  a small random perturbation of the  energy functional
$$ [u]^2_{H^s( \La, \R^d)} + \int_\La W(u(x)) dx  $$
for $s \in (0,1),$ where  the non-local part  $ [u]^2_{H^s( \La,\R^d)}$ denotes the total contribution  from $\La \subset \R^d$  in  
the $H^s (\R^d)$  Gagliardo semi-norm of  $u$  and $W$ is a double well potential. We show that  
 there exists, as $\La $ invades $ \R^d$,  for almost all  realizations of the random  term  a  
 minimizer  under compact perturbations, which is unique when $d=2$,  $s \in (\frac 12,1)$  and when  $d=1$, $s \in [\frac 14, 1).$ This uniqueness is a consequence of the randomness.   When the random term is   absent,  there are two minimizers which are  invariant under translations  in space, $u = \pm 1$.  
 \end{abstract}

\maketitle

\section{Introduction}
Non local functionals, related to fractional Levy partial differential equations,
appear frequently in many different areas of  mathematics and find many applications in engineering, finance \cite {RP}, physics \cite {MK},  chemistry \cite {BS}  and biology \cite {WGMN}.
We consider     non local functionals     representing   
 the free energy of a material with
two (or several) phases, see \cite {D},  on a 
a scale,    the so-called mesoscopic scale,  which is much larger than the atomistic  scale
so that the 
adequate description of the state of the material is by a {\em continuous}
scalar order parameter $m:\ \La \subseteq\R^d\to \R$.    
The  minimizers of these functionals are functions $m^*$ 
representing  the states  or  phases  of the materials. 

The natural question that we pose is the following:  What  
happens  to these minimizers  when   an  external, even very weak,  
random  potential   is added to the deterministic functional?
Does the   number of minimizers remain the same, i.e will  
the material always have the same number of states (or phases)?
Is there some significant difference in the qualitative properties of 
the material  when the randomness is added?  These are standard questions 
in  a calculus of variations framework.

Partial answers to these type of questions  were recently given  in  two  papers by the authors in the context 
of  the Ginzburg Landau functional, i.e in the case where the interaction  energy is    local and it is  modelled by $  \langle m, (-\Delta) m \rangle$  there $\langle \cdot, \cdot \rangle$ stands for  the $L^2$ scalar product and    $m$  is  taken in a  function  space which makes  the  scalar product   finite,  see \cite {DO} and  \cite {DO2}.
Here we  consider  a   functional in which the interaction energy is   non local,    i.e.  the state of the material at site $x \in \La$    depends on the state of  the material   in   all   $\R^d$.    We model this non local interaction energy  using the fractional Laplacian. 

 This nonlocality of the interaction needs a very different approach compared to \cite{DO} and \cite{DO2} because of the suitable interpretation of "boundary condition" in the case of a long-range interaction. In particular, an extensive 
part of the analytical work in the present paper is devoted to  so-called minimizers under compact perturbations, see    Definition \ref {min0}.

The {\em interaction energy} is given by 
 $  \langle m, (-\Delta)^s m \rangle$   for  $0<  s<1,$    the scalar product and the function space  for $m$ need to be suitable defined.   In the extreme case $s=1$ one gets the  Ginzburg Landau interaction energy and when $s=0$ one gets    $(-\Delta)^s =I$ where $I$ is the identity operator,  so $m$ at site $x$  interacts only with itself.    
 
 We add to this non local interaction energy  which penalizes spatial changes in $m$
a {\em double-well potential} $W(m),$
i.e. a nonconvex  function which has exactly two minimizers,
for simplicity $+1$ and $-1,$ modelling a two-phase material.

Finally, we add
 a term which couples $m$ to a {\em random field}  
$\theta g(\cdot,\omega)$ with mean zero,
variance  
$\theta^2$ 
and unit correlation length; i.e a term    
which   prefers at
each point in space one of the two minimizers of $W(\cdot)$ and thus 
breaks the translational invariance, but is
"neutral" in the mean.  

A    functional   with the aforementioned
properties is the following  functional 
\begin{equation} \label{functionalA}
G^{m_0}_1 (m,\omega,\La )= [m]^2_{H^s( \La,\R^d)} +
 \int_{\La} W(m(x)) \rm {d }x
-  \theta \int_{\La}  g_1 (x,\omega) m(x)\rm {d }x,
\end {equation}
  where 
\begin{equation} \label{J1}[m]^2_{H^s( \La,\R^d)}  =    \int_{\La} \rm {d }x \int_{\La} \rm {d }y\frac { | m(x) - m(y)|^2} {|x-y|^{d+2s}}  +  2   \int_{ \La } \rm {d }x \int_{\R^d \setminus \La} \rm {d }y\frac { | m(x) - m_0(y)|^2} {|x-y|^{d+2s}} \end {equation}
   denotes the total contribution from $\La$ to the $H^s (\R^d)$  Gagliardo semi-norm of  $m,$
if we set  $m= m_0$ in $\R^d \setminus \La$ in \eqref {J1}.
The Gagliardo semi-norm is given by 
 \begin{equation} \label{sp1}  \begin {split} &   \int_{\R^d} \rm {d }x \int_{\R^d} \rm {d }y\frac { | m(x) - m(y)|^2} {|x-y|^{d+2s}} =    [m]^2_{H^s( \La,\R^d)}  +   \int_{\R^d \setminus \La } \rm {d }x \int_{\R^d \setminus \La} \rm {d }y\frac { | m_0(x) - m_0(y)|^2} {|x-y|^{d+2s}}  . \end {split}
\end {equation} For the minimization problem   the term depending only on the value  of  $m_0$  in  the Gagliardo semi-norm    is irrelevant, since this term is kept fixed trough the minimization procedure.  
For dimensional reason the right hand side of \eqref {J1} should be multiplied by 
    $  c_{d,s} $,  a   normalizing constant  which  degenerates  when $ s \to 1$ or $s \to 0$.  
      In the following   the constant  $c_{d,s}$  does not play any role, so we  replace it by  $1$. 

      We are interested  in      determining    the  {\it macroscopic minimizers}    of  \eqref {functionalA},  i.e   minimizers   of  \eqref {functionalA}  over sequences of regions   $\La_n$    so that $\La_n\nearrow \R^d$   as $n\to \infty.$   Namely for any given  $ \La$ and fixed boundary value $m_0$ the minimizers of  \eqref {functionalA} over any  reasonable set of functions will depend on the boundary value $m_0$.   Physically one is interested in  taking $ \La$ large enough and to characterize   the minimizers  in a region deep inside  $ \La$ and detect if, even so deeply inside, the boundary condition  is  felt.   
In  other word $ \La$ needs to be  large to invade $\R^d$ and we are interested in characterizing the macroscopic minimizer which we construct   by a limit procedure   using minimization on
a sequence of finite subsets of $\R^d$. 

When $\theta=0$, i.e without random term, 
the constant functions equal to $\pm 1$ are the two macroscopic minimizers:
One can obtain  the $+1$ ($-1$) minimizer  as the limit of the minimizers of  \eqref {functionalA} when $\theta=0$ with  strictly  positive  (strictly negative) boundary values   by making use of the fact that the cost of a "boundary layer" near the boundary of large balls is of smaller order than the volume as the balls invade $\R^d,$ a point to which we will come back below, see \eqref{G10}. 

 When the random field is added, the 
constant functions  equal to $\pm 1$ are not minimizers anymore, due to the presence of the random fields.
 The question  is to show whether there are  still  two macroscopic  minimizers, each one close in some topology to  the constant minimizers $1$ and $-1$.
 
We are able to show   in $d=2$ for  $s \in (\frac 12 ,1)$ and in  $d=1$  for  $  s \in [\frac 14,  1) $ that for almost all the realizations of the randomness,  there  exists  one    macroscopic minimizer   which is unique under compact perturbations.  In this regime the  boundary conditions is not felt by the minimizer. 
 This is an example of uniqueness  induced by  random terms. The uniqueness holds only in the limit
$\La \nearrow \R^d$ and is sensitive to the type of randomness added.
  We will come back to this point   in subsection 2.1.  For    values of  $d$ and $s$  different from the ones for which we state the uniqueness  result  we expect,  for almost all the realizations of the randomness,  the existence of at least two macroscopic
minimizers, one "close" to the constant minimizer $1$, the other  "close'' to   the constant minimizer $-1$. But this issue  is still open. 
The strategy of our proof is   based on  the following   steps.  
We   prove  first that  for almost all the realizations of the random field  
there exist two  {\it macroscopic extremal  }
minimizers  $v^\pm (\cdot, \om)$ 
  so that any other macroscopic minimizer under compact perturbations 
$u^\ast$  
satisfies $v^- (\cdot, \om) 
\le  u^* (\cdot, \om) \le v^+ (\cdot, \om)$. 
 This construction  requires two limit procedures.
First,    for any  bounded, sufficiently regular   subset of $\R^d$,   $\La$, and for any  $K>0$
 we  determine the  minimizers of $G_1$ in $\La$ with boundary condition $v_0=K$.
Since the functional is not convex there might be many minimizers. 
 Because the set of minimizers in a bounded domain $\La$ is ordered and compact,  we can  single out one specific minimizer which we call  the maximal $K-$ minimizer. 
 Similarly we  single out one specific minimizer $G_1$ in $\La$ with $v_0=-K$ boundary condition,
which we call  the minimal $K-$ minimizer.   The   maximal $K-$ minimizer and the  minimal $K-$ minimizer of $ G_1$ in $\La$  have the property that any other minimizer of $G_1$ in $\La$ with boundary condition $v_0$, $\|v_0\|_\infty \le K $  is  point wise smaller  than   the   maximal $K-$ minimizer and larger than  the    minimal $K-$ minimizer of $ G_1$ in $\La$.
Then we let $\La$ to invade $\R^d$ obtaining  two    infinite volume functions $u^{\pm, K}$,
and we show that  they  are  infinite volume  minimizers   under compact perturbations of $G_1$.
At last, we define  $v^\pm (\cdot, \om)$ as the point wise limit as $K\to \infty$  of  $u^{\pm, K}$, proving again  
  that   $v^\pm (\cdot, \om)$  are  extrema infinite volume  minimizers  under compact perturbations.
 Then  we  show that  for any  $  s \in   (0,   1) $ there exists a   positive constant $C$, so that  for any  bounded, sufficiently regular 
$\La \subset 
\R^d$   ,   for almost all the realizations of the random field,  
\begin {equation} \label {G10}   
\left  | G_1^{v^{+}}  ( v^+,   \om, \La )-  
G_1^{ v^-}  (v^-,   \om, \La  ) \right |   \le C    |\La|^{\frac {d-1}d}  \1_{\{s\in (\frac 12, 1)\}}  + C  |\La| ^{\frac {d-2 s} d}   \1_{\{ s\in (0, \frac 12)\} } +    \1_{\{s= \frac 12\}}  |\La|^{\frac {d-1}d}  \log |\La|. \end {equation}
 The minimizers    $v^\pm (\cdot, \om)$ depend in a highly non trivial way 
on the  random fields $ \{g(x,\omega)\}_{\{ x \in \Z^d\}}$.  Therefore also the  
difference $G_1^{v^{+}}  ( v^+,   \om, \La)-  G_1^{ v^-}  ( v^-,   \om, \La)$ 
depends on the  random fields in all of $\Z^d$. 
We  take a sequence $\La_n \subset \La_{n+1}$
and we  show that,  conditioning on the random fields in     
$\La_n$ (i.e taking the expectation  over only the random fields outside $\La_n$)
$$   F_n (\om):=    \E \left [   \left \{G_1  ( v^+ (\om),   \om, \La_n )-  
G_1  ( v^-(\om),   \om, \La_n ) \right \}| \BB_{\La_n} \right ]  $$  has significant 
fluctuations, with variance of the order of the volume.
Here $\BB_{\La_n}$ is  the $\s$ algebra generated by the random field 
in  $\La_n$. 
Namely we show that
$$  \E  \left [   F_n (\cdot )\right ] =0, $$
and for  $t \in \R$
 \begin{equation} \label{mars3}
 \liminf_{n \to \infty} 
\E  \left [ e^{t  \frac {F_n} {\sqrt { \La_n}} } \right ] 
\ge e^{\frac {t^2 D^2}  2},   \end{equation}   
where $D^2$ is given in \eqref{may1}. 
This holds in all dimensions and for all $ s \in (0,1)$.   
 In  $d=1$  and for $s \in  [\frac 1 4, 1)$,  in $d=2$   and for    $s \in   (\frac 1 2, 1)$  the bound  \eqref {mars3} generates a contradiction  with the  
bound \eqref  {G10},     unless $D^2=0$.  
 When   $ D^2 =0$  we show that  
  $ M=  \E [ \int_{Q(0)}  v^+] -   \E [  \int_{Q(0)}  v^-] =0$.   
Further, we  show  that   point-wise $v^+  \ge v^-$,   
therefore  $ \E [ \int_{Q(0)}  v^+] =   \E [  \int_{Q(0)}  v^-] =0$
and $v^+(\cdot,\om)= v^-(\cdot, \om)$,  for almost all realizations of the random field.  
The probabilistic  argument  has been already applied   by
Aizenman and  Wehr, \cite {AW}, in the context of 
Ising spin systems with random external field,
see also the   monograph  by Bovier, \cite{B}, for a survey on  this  subject.

 It is instructive to  understand what one can say about the functional \eqref  {functionalA} when $ \theta=0$.   
 Denote   $J^{m_0} (m, \La) $  the functional \eqref {functionalA} when $ \theta =0$.  In this case 
the constants  $m(x)=  \tau $  for $x  \in  \R^d$ and  $ \tau = \pm 1$ are 
 the   only {\em bounded} global  macroscopic minimizers under compact perturbations.   
 To pass to    a so-called {\em macroscopic} scale, which
is coarser than the mesoscopic scale,       we 
rescale space with a small parameter
$\eps.$ If $\D=\eps \La $ and $u(z)=m(\eps^{-1}z) $  and  $u_0(z)=m_0(\eps^{-1}z)$  we obtain
\begin {equation} \label {functB1}
\tilde J^{u_0}_\eps  ( u,  \D )=    \e^{2s-d} [u]^2_{H^s( \D,\R^d)} +
 \e^{-d} \int_{ \D} W( u(z)) \rm {d }z. 
\end {equation}

 Functionals with a  finite energy on this scale must be Lebesgue almost everywhere close to one of the two minimizers. 
The second step is  to determine  the cost of forming an interface between the spatial regions occupied by these two different minimizers. 
 
As in the case of the corresponding local functional  one needs to  normalize  $\tilde J^{u_0}_\eps  ( u,   \D)$ by  a power of $\epsilon$ related to  
the   dimension of the interface,  which is not necessarily an integer in this case, see also Lemma \ref{boun1}.  Computations similar to the ones done to obtain \eqref  {G10} give  for  $\theta=0$ 
 a factor of     $ \e^{-d+1}$ when $ s \in (\frac 12, 1)$,   $ \e^{-d+2s}$ when $ s \in (0, \frac 12)$,
and by  $ \e^{-d+1} \log {\frac 1 \e} $ when $ s = \frac 12$. 
Therefore   
we obtain
\begin {equation} \label {functB}
J^{u_0}_\eps  ( u,   \D)=  \left \{  \begin {split}  &     \e^{2s-1}  [u]^2_{H^s( \D,\R^d)}  +
  \e^{-1} \int_{\D} W( u(z)) \rm {d }z, \quad  s \in (\frac 12, 1)\cr &
      [u]^2_{H^s( \D,\R^d)}  +
 \e^{-2s}  \int_{\D} W( u(z)) \rm {d }z, \quad  s \in (0, \frac 12)  \cr &
   \frac { \e^{2s} } { \e \log \e }  [u]^2_{H^s( \D,\R^d)}   +
 \frac 1 {\e \log \e }  \int_{\D} W( u(z)) \rm {d }z, \quad  s = \frac 12.  
 \end {split} \right. 
 \end {equation}  
The $\Gamma$-convergence  for the  functional \eqref {functB} has been studied  by   Savin and Valdinoci,  \cite {SV}. 
They show  that  the functional $ J^{u_0}_\eps  ( u,   \D)$ $\G-$ converges to the classical minimal surface functional when $s \in [\frac 12,1)$  while, when $s \in (0,\frac 12)$ the functional $\G-$ converges to the nonlocal minimal surface functional. 
There are in the literature other results dealing with    $\G-$ convergence of non local functionals, see e.g. \cite {GM}, \cite {GP}, \cite {Go} and references therein, but they are different from  the deterministic part of the functional that we are considering, either for the explicit form  or  because they  do not consider the full interaction    of $\La $  with all of   $\R^d$.    Physically this implies that  the particles in the domain $\La$ interact with all the particles in $\R^d$ and not only with those ones in $\La$,  i.e. a sort of nonlocal Dirichlet boundary condition.

\section{Notations and Results}

 We denote by $\La \subset \R^d$ a generic open subset of  $\R^d$,  by $ \partial \La$  the boundary of $\La$ and by $\La^c= \R^d \setminus  \La$.  When $\La$ is a bounded subset of $\R^d$ we  write  $\La \Subset \R^d$.  We denote by 
     $ |x|$ the euclidean norm of $x \in \R^d$,  by $ | \La|  $ the  volume of $\La$, 
  by $ \operatorname {diam} (\La) = \sup  \{  |x-y|,   \quad   x \  \hbox  {and}  \ y \in  \La \}$ 
  and by $ d_{\partial \La}(x)$    the  euclidean  distance from $x$ to $ \partial \La$.
 We will consider  domain  $\La$  with Lipschitz boundary  regularity,  i.e the boundary  can be thought of as locally being the graph of a Lipschitz continuous function, see for example \cite {DPV1}.
 { It is useful to introduce the following definition.  We  say that    
  a set with Lipschitz boundary  $\Lambda  \Subset \R^d $ is {\em cube-like}  if  $ {\mathcal H}^{d-1}(\partial \Lambda)\le C|\Lambda|^{\frac{d-1}{d}}$ and ${\rm diam}(\Lambda)\le C|\Lambda|^{\frac{1}{d}},$ where   ${\mathcal H}^{d-1}$ is  the $d-1$-dimensional Hausdorff measure and $C>0$ is a constant depending only on the dimension $d$.

   For $t $ and $s$ in  $\R$  we denote $s\wedge t= \min \{s,t\}$ and $s \vee t =  \max \{s,t\}$.
  For  $\La \subset \R^d$, we denote by $C^{k, \alpha} (\Lambda)$, $ k \ge 0$ an integer, $\alpha \in (0,1]$
the set of functions continuous  and  having continuous derivatives up to order  $k$,     such that the  $k$-th  partial derivatives  are H\"older continuous  with exponent  $\alpha$.  
 \subsection{The disorder} 
 The disorder or random field is constructed with the help of
a family  of 
independent,
identically distributed  random variables with mean zero and variance  equal to 1.
  We assume  that each  random variable    has   distribution    
absolutely continuous with respect to the
Lebesgue measure 
 and that the  Lebesgue density  is a 
symmetric, compactly supported function on $\R$. 
The corresponding infinite product measure on 
${\R}^{\Z^d}$   will be denoted by 
$\Pr$ and by $\E[\cdot]$ the mean with respect to $\Pr$.
We denote this family of random variables by
$\{g(z,\omega)\}_{z\in \Z^d}$, $ \omega \in \Omega$   where we identify  $\Omega $ with  ${\R}^{\Z^d}$. 
These assumptions  imply  that there exists  a finite $A>0$ so that 
  
  \begin {equation} \label {eq:ass}  
\quad \E[ g(z)]=0, \quad 
\E [g^2(z)]=1,   \quad  \forall z \in \Z^d \quad \text{and}\quad      \|g\|_\infty = \sup_{z} |g(z, \om)|  = A, \quad  \Pr\text{ - a.s.}  \end {equation}
 The boundedness  assumption is  not essential. 
Different  choices 
of $g $ could be handled by minor modifications provided  $g$ 
is  still a random field with finite correlation length, 
invariant under (integer) translations and such that
$g(z,\cdot)$ has a symmetric distribution,  absolutely continuous w.r.t 
the Lebesgue measure and  $ \E[ g(z)^{2+ \eta}] < \infty$, $ z \in \Z^d$ for $\eta >0$. 
The  method  does not apply  when   $g$ has atoms,  i.e. its distribution is not absolutely continuous with respect to the Lebesgue measure, 
see Remark  \ref {AC1}.    It is not clear to us if this requirement  is   purely technical   or if the discrete distribution of the random field may cause a degeneracy of the ground state like in the Ising spin systems  \cite{AW}. 

The symmetry of  the measure  
$\Pr $  is essential for obtaining the result. 
Namely if  $\Pr $ does not have a symmetric distribution, 
it would be no longer natural to compare the  
qualitative properties of the functional  \eqref{functionalA}
for $\theta \neq 0$  with the functional \eqref{functionalA} with $\theta=0.$
Therefore   
in the following we  always assume   that $\Pr$ is symmetric.

We denote by  $ \BB$ the product $\sigma-$algebra and by $ \BB_\Lambda$, $  \Lambda \subset \Z^d$,  the $\sigma-$ algebra generated by $\{ g(z, \om): z \in \Lambda \}$.  
In the following we  often identify  the random field $ \{g(z, \cdot): z \in \Z^d\}$ with the coordinate maps  $ \{ g(z, \om)= \omega (z): z \in   \Z^d\}$. 
  To use ergodicity properties of the random field it is convenient 
to  equip  the  probability space $ (\Om, \BB, \Pr) $    with some extra structure.
First, we define the action $T$  of the translation group  ${\Z}^d$  on $ \Om$. We will assume that $ \Pr$  is invariant under this action and that the dynamical system  $ (\Om, \BB, \Pr, T)$ is stationary and ergodic.
In our model 
the action of  $T$  is  for $y \in \Z^d$  
 \begin {equation} \label {parisv1} (g (z_1,  [T_y \om]), . . . , g (z_n,  [T_y \om])) = ( g ( z_1+y,  \om) , . . . ,  g (z_n + y,  \om )). \end {equation}
The disorder or random field in the functional will be obtained
 setting for  $x \in \Lambda$  
 \begin{equation}\label{gfrombern}
g_1 (x, \omega):=\sum_{z\in \Z^d}g(z,\omega)
\1_{ (z+[-\frac 12 ,\frac 12 ]^d)\cap \Lambda }(x),
\end{equation} where for any Borel-measurable set $A$
$$
  \1_A (x):= \begin{cases} &
1, {\rm if\ } x
\in A\\ &  0 \ {\rm if\ } x \not\in A . 
\end{cases}
$$
\subsection{The double well potential}  Next we  define  the   ``double-well potential''  $W$: 
\vskip0.5cm
\noindent
{\bf  Assumption (H1) }  $ W \in C^{2}(\R)$,  $W\ge 0$, $ W(t)=0$ iff $t \in
\{-1,1\}$, $W(t) =
W(-t)$ and
$W(t)$ is strictly decreasing in   $[0,1]$. Moreover there exists
$\delta_0$ and $C_0>0$ so that
\begin{equation} \label{V.1}   W(t) = \frac 1 {2 C_0} (t-1)^2 
\qquad \forall t \in (1-\delta_0,  \infty).
\end{equation}
 Note that $W$ is slightly different from the standard choice
$W(u)=(1-u^2)^2.$ Our choice simplifies some proofs because it makes
the Euler-Lagrange equation  linear provided solutions stay in one ``well.''
Note that in order to obtain our uniqueness result we could replace the equality in (\ref{V.1}) by a lower bound on $W(t)$ of the same form. 
 \subsection{The functional}  We start introducing  the functional spaces in which we  define
 the nonlocal interaction term. 
 \begin {defin}  {\bf Fractional Sobolev spaces} 
  Let $ D \subset \R^d$ be an open domain and   $s \in (0,1)$.  We define the  fractional Sobolev space  $H^s (D)$  as the set of functions $ f \in L^2(D)$ so that
  $$\int_{D \times D }  \frac {( f(x) - f(y))^2 } { |x-y|^{d+2s}  }d x dy  < \infty.$$ 
 This space, endowed with the norm 
 $$ \|f\|_{H^s (D)}=   \|f\|_{L^2 (D)} +\left (  \int_{D \times D }  \frac {( f(x) - f(y))^2 } { |x-y|^{d+2s}  }d x dy \right)^{\frac 12}$$
 is an  Hilbert space. 
 We will say that $ f \in H^s_{loc} (\R^d)$, $s \in (0,1)$, if  $ f \in H^s (B_R)$   for any ball of radius $R$ in $\R^d$.
 \end {defin}  
   
   \medskip

For  $ v     \in   H^s_{loc} (\R^d) $, 
$\Lambda \Subset \R^d $ 
  denote       
   \begin{equation} \label{functional2}
\KK_1 (v,\omega,\Lambda) =     \int_{\Lambda} \rm {d }x \int_{\Lambda} \rm {d }y\frac { | v(x) - v(y)|^2} {|x-y|^{d+2s}} +
 \int_{\Lambda} W(v(x)) \rm {d }x
-  \theta \int_{\Lambda}  g_1 (x,\omega) v(x)\rm {d }x.
\end {equation}
Now we introduce some definitions  needed  to specify ``boundary conditions" in a sense appropriate for nonlocal functionals.  \newline 
For  any $ \Lambda \Subset \R^d $ and $\Lambda_1 \subset \R^d$, $  \Lambda_1 \cap  \Lambda= \emptyset$,   for $v$ and $u$ in     $H^s_{loc} (\R^d)$ denote 
\begin{equation}  \label {int1}
 \WW ( (v,  \Lambda),  ( u, \Lambda_1))= 2  \int_{\Lambda} \rm {d }x \int_{  \Lambda_1} \rm {d }y
 \frac { | v (x) - u(y)|^2} {|x-y|^{d+2s}}   
  \end {equation}
  the interaction between    the function $v$   in   $\Lambda$   and the function $u$ in  $\Lambda_1$.
  Note that if  $\Lambda_1$ is not a bounded set,  the  term in \eqref {int1} might not be finite. 
  We will show  in Lemma  \ref {boun1} that when $v \in     H^s_{loc} (\R^d) \cap L^\infty (\R^d)$ then
 $ \WW ( (v,  \Lambda),  (v, \Lambda_1))$  is  bounded, the bound depends on $| \La|$. 
   When   $\Lambda_1=   \Lambda^c$  \ and $u=v$  we  simply write 
\begin{equation}  \label {fe1} \WW (v, \Lambda) = 2   \int_{\Lambda} \rm {d }x \int_ {\Lambda^c} \rm {d }y\frac { | v(x) - v(y)|^2} {|x-y|^{d+2s}}.    \end {equation}
    
     \begin {defin}  {\bf  The Functional}  For  any $ \Lambda \subset \R^d $, $ v     \in   H^s_{loc}(\R^d)   \cap L^\infty (\R^d) $   we define 
\begin{equation}  \label {funct11}
  G_1(v,\omega,\Lambda ) =\KK_1 (v,\omega,\Lambda)+  \WW (v, \Lambda).  
    \end {equation}
    Whenever we need to stress the dependence of $G_1$ on the  value of   $v$ outside $\La$, i.e.   $v (y) = v_0(y), y \in \La^c$,  
    we   will  write 
    \begin{equation}  \label {funct110}
  G_1^{v_0}(v,\omega,\Lambda ) =\KK_1 (v,\omega,\Lambda)+    \WW ((v, \La) (v_0, \La^c)). \end {equation}
  \end {defin}  
 We  list some useful properties of the functionals  $G_1$ and  $\KK_1$ that follow immediately from the definitions. 

\begin{lem}{\quad }

\begin {itemize}
\item $\KK_1$ is superadditive, i.e. if  $A$ and $B$ are disjoint sets then
$$ \KK_1(v, \om, A \cup B) \ge  \KK_1(v, \om, A )+ \KK_1(v, \om,  B),$$
\item 
$G_1$ is  subadditive, i.e. if  $A$ and $B$ are disjoint sets then
\begin{equation}  \label{rome3} G_1(v, \om, A \cup B) \le   G_1(v, \om, A )+ G_1(v, \om,  B). \end {equation}
\end {itemize}
\end{lem}
 
    \begin {defin}  \label {min0}{\bf The minimizers }   \begin{enumerate}\item 
    We say that $u \in H^s_{loc} (\R^d)\cap L^\infty(\R^d)$ is a minimizer  under compact perturbations  for $G_1$ in   $ \Lambda \subset \R^d $ if for any compact  subdomain $ U \subset \La$ we  have
    $$   G_1(u,\omega, U )< \infty, \quad \Pr\  a. s. $$
    and 
    $$  G_1(u,\omega,U )\le   G_1(v,\omega,U )  \quad \Pr \ a. s.$$
    for any $v$ which coincides with $u$ in $ \R^d \setminus U$.
    \item  Let $v_0\in  L^\infty(\R^d)$ be independent of $\omega\in \Omega.$  We say that $u \in H^s_{loc} (\R^d)\cap L^\infty(\R^d)$ is a $v_0$- minimizer \ for $G_1$ in   $ \Lambda \subset \R^d $ if for any compact subdomain $ U \subset \La$ we  have
    $$   G_1^{v_0}(u,\omega, U )< \infty, \quad \Pr \,a. s. $$
    and 
    $$  G_1^{v_0}(u,\omega,U )\le   G_1(v,\omega,U )  \quad \Pr\, a. s.$$
    for any $v$ which coincides with $v^0$ in $ \R^d \setminus U$. 
    \item  We say $u$ is a free minimizer on $\Lambda$ if it minimizes ${\mathcal K}_1(\cdot, \omega,\Lambda)$ in $H^s(\Lambda).$  
\end{enumerate}
    \end {defin}
    Note that $v_0$ will usually be a constant function.
    \medskip
    \begin{rem}[Existence]
   Existence of $v_0$-minimizers  (for sufficiently regular $v^0$)  and free minimizers in a bounded  Lipschitz set $\La \subset \R^d$
    follows from the compact embedding  of $H^s (\La)$ in $L^2(\La)$ and the lower semicontinuity of the $H^s$-norm.  We prove  the existence of a $v_0$-minimizer in Lemma \ref {jt1} and  Lemma \ref {jt1a} in the Appendix. 
    The existence of  exactly one 
    minimizer under compact perturbations is a consequence of the main theorem. 
    \end{rem}  
  \medskip
    \begin {defin}  { \bf   Translational covariant states}  We say that  the function  $ v: \R^d \times \Omega \to \R$  is   translational covariant if
 \begin {equation} \label {rome1}  v(x+y, \omega) = v(x, [T_{-y} \om]) \quad   \forall y \in \Z^d, \quad x \in \R^d. \end {equation}
\end {defin}

Our   main result is  the following. 
 \begin{thm} \label{min1}  Take    $ d=2 $ and  $ s \in  (\frac 1 2, 1)$  or $d=1$ and $ s \in [\frac 14,1)$,   $ \theta $    strictly positive. Let    $n \in \N$, $  \La_n = (-\frac n 2, \frac n 2)^d$ \begin {footnote}     {One could take any  increasing,  {\it  cube-like},  sequences of sets    $ \{\La_n\}_n$, $  \La_n \subset  \R^d$    invading $\R^d$. The proof goes in the same way. } \end {footnote},  $v_0 \in L^\infty(\R^d)$     and   $ u^*_n$  be   a $v_0$-minimizer of $G_1$ in $\La_n$  according to Definition  \ref {min0}. 
 Then   $\Pr$ a.s.   there exists  a {\em unique}  $u^* (\cdot, \om) ,$ {\em independent of the choice of $v_0,$}    defined as  
 \begin {equation}  \label {ra1}\lim_{n \to \infty}   u^*_n (x, \om)= u^* (x, \om) \qquad {\rm (uniformly\ on\ compacts\ in}\ x{\rm)} \end {equation}
   so that  
   \begin{itemize}
 \item  $u^* (\cdot, \om) $ is   translation covariant, see \eqref {rome1}. 
 \item   $\| u^* (\cdot, \om)\|_{\infty}   \le 1+ C_0 \theta \|g \|_\infty $  where $C_0$ is the constant in \eqref {V.1}.  
  \item     $u^* (\cdot, \om)  \in C^{0,\alpha}_{\operatorname {loc}  (\R^d}) $ for any $\alpha < 2 s$ when    $ 2s \le 1$,    
  $u^*  \in  C^{1,\alpha}_{\operatorname {loc}} (\R^d) $  for any $ \alpha < 2s-1$,  when  $ 2s >1$.
     \item    $$        \E[u^*(x ,\cdot) ]=0, \qquad  \forall x \in  \R^d .$$ 
  \end {itemize}
  \end{thm}

   \begin {rem}  \label {ra2}    Since  for any set $\La \Subset \R^d$, $C^{0,\alpha} (\La) \subset C^{0,\beta}(\La)$  for $ \beta <\alpha$  and the inclusion is compact,    the convergence  in \eqref {ra1} holds  in $C^{0,\beta}$,  $ \beta < 2s$ when $s \in (0, \frac 12 ],$ because we can find $\alpha$ with $ \beta < \alpha<2s.$    Similarly one obtains convergence of  \eqref {ra1} in  $C^{1,\beta}$,  $ \beta < 2s-1$ when $s \in (\frac 12,1)$.  \end {rem}

\begin {rem}  When  $\theta=0$ in \eqref {functional2},  i.e the random field is absent,   the minimum value of $\mathcal {K}_1(\cdot,\cdot,\Lambda)$ is zero for any bounded $\Lambda$ and 
there are  exactly two translation covariant minimizers under compact perturbations, the constant   functions  identically 
equal to $1$ or to $- 1$. 
\end {rem}

    \section{Finite volume Minimizers}
 In this section we   state      properties for minimizers of the  functional $ G_1$ in any bounded set   
  $ \Lambda \subset \R^d $.
    These properties hold in all   dimensions $d$, for all bounded $\Lambda$ \ with Lipschitz boundary   and for    every $ \om \in \Omega$.    
The $\om$ plays the role of a parameter.     
 We start  showing  that   
to determine the   minimizers of  $ \KK_1$ in $\La$
it is sufficient to consider    functions   $v$ 
 satisfying  a uniform $L^\infty$-bound.

  For any $t>0$ denote  by $v^t=   t \wedge v\vee (- t)$.  
 \begin{lem} \label{A} 
Let the double well potential $W$ satisfy Assumption  (H1).  
 \begin{enumerate}\item For all $\om \in \Omega$,  for all 
$v \in  H^s  (\La)$ and all $t  \ge 1+ 
C_0\theta  \|g \|_\infty$ 
\begin{equation}\label{eqabove}
   \mathcal K_1(v, \om, \La)-  \mathcal K_1 (v^t,\om, \La )\ge  
  \int_{\Lambda_t}   \left (   C_0^{-1} (t-1)  -
 \theta \|g\|_\infty
\right ) (|v(y)|-t),  
\end{equation}
where $C_0$ is the constant in \eqref{V.1} and 
$ \Lambda_t= \{ y \in \Lambda: |v(y)| >t\}.$
\item  Take  $ v_0 \in H^s_{loc} (\R^d) \cap L^\infty (\R^d)$ and  $t\ge \max\{\|v_0\|_\infty, 1+ 
C_0\theta   \|g \|_\infty\} $.   The result stated in \eqref {eqabove}   holds for $G_1^{v_0}(v, \om, \La)$. This implies in particular that minimizers of $G_1^{v_0}$ are bounded uniformly by
$\max\{\|v^0\|_\infty, 1+ 
C_0\theta   \|g \|_\infty\}.$  
\end{enumerate}
\end{lem}
   
\begin{proof} 
We have that  for  $x $ and $y$ and any function $v$ and $w$   
$$  [v (x) - w(y) ]^2 \ge [ v^t(x)- w^t(y)]^2 .$$
We  immediately obtain 
\begin{equation*} 
  \begin{split}   &    \mathcal K_1  (v,\om,\La)-
   \mathcal K_1 (v^t,\om, \La )  \ge
    \int_{\Lambda_t}   \left (   W(v(y))- W(t) \right )
dy
 -  \theta \int_{\Lambda_t} dy   g_1 (y,\om) [
v(y)- \operatorname{sign} (v(y)) t], 
 \end{split}
\end{equation*}
 and from  Assumption (H1) and the $L^\infty$-bound on $g$ we derive (\ref{eqabove}).
  The proof of (2) is   a   consequence of (1)  by choosing  $t\ge \max\{\|v_0\|_\infty, 1+ 
C_0\theta   \|g \|_\infty\}$.
  \end{proof}
  \vskip0.cm 
 Next we show  that  the functional  \eqref {funct11}   is  finite   when  $v \in    H^s_{loc}(\R^d) \cap L^\infty (\R^d)$.  To this aim it is sufficient to show   that $  \WW (v, \Lambda)$, defined in \eqref {fe1}, is finite. 
  \begin {lem} \label {boun1} Let  $v \in    H^s_{loc}(\R^d) \cap L^\infty (\R^d)$,    $\Lambda \Subset \R^d$ 
 and 
   $C=  C ( \|v\|_\infty, d, s)$  be a  generic constant   which might change from one occurrence to the other.  
    Suppose that  $\Lambda$ is {\em cube-like}.\footnote{ The lemma holds for Lipschitz domains $\Lambda$, but then the generic constant $C$ depends on the shape of the domain.} Then  we have 
  \begin{equation} \label {t1} \WW (v, \Lambda) \le C      | \Lambda|^{\frac  {d-2s} d}, \qquad s \in (0, \frac 12).      \end {equation}
  When $ s \in  [\frac 12, 1)$  denote by $B_{1} (\partial \Lambda)= \{ x \in \R^d: d_{\partial \La}(x)  \le 1\}$   we have 
  \begin{equation} \label {t3} \WW (v, \Lambda) \le    \|v\|_{H^s(B_1 (\partial \Lambda))} + \left\{\begin{array}{cc}
 C        | \Lambda|^{\frac  {d-1} d},  &\qquad  s \in  (\frac 12, 1),\\
 C        | \Lambda|^{\frac  {d-1} d}\log(|\Lambda|),  &\qquad  s =\frac{1}{2}.
\end{array}\right.
   \end {equation}
   When $ s \in  [\frac 12, 1)$  and  $v   \in C^{0, \alpha} ( B_{1} (\partial \Lambda)) $  for $\alpha > s -\frac 12 $
 \begin{equation} \label {t3c} \WW (v, \Lambda) \le     
 \left\{\begin{array}{cc}
 C        | \Lambda|^{\frac  {d-1} d},  &\qquad  s \in  (\frac 12, 1),\\
 C        | \Lambda|^{\frac  {d-1} d}\log(|\Lambda|),  &\qquad  s =\frac{1}{2}.  
\end{array}\right.
   \end {equation}
       \end {lem}
  \begin {proof}
For any   $ s \in (0, \frac 12)$ we have 
\begin{equation}   \label {oct2}\begin {split}& \int_{\Lambda} \rm {d }x \int_ {\Lambda^c} \rm {d }y\frac { | v(x) - v(y)|^2} {|x-y|^{d+2s}} \cr & \le C  \int_{\Lambda} \rm {d }x \int_ {\Lambda^c} \rm {d }y\frac {1} {|x-y|^{d+2s}} \le   C   \int_{\Lambda} \rm {d }x \int_ { \{ y\in \R^d:\ |x-y|\ge  d_{\partial \La}(x)  \} }\frac {1} {|x-y|^{d+2s}} \rm {d }y  \le C  \int_{\Lambda}  |d_{\partial \La}(x)|^{-2s} \rm {d }x  \cr & \le C({\rm diam}(\Lambda))^{1-2s}{\mathcal H}^{d-1}(\partial \Lambda)
 \le C |\La|^{\frac {d-2s} d}.
 \end {split} \end {equation} 
  Note that for cubes ${\rm diam} (\Lambda)\le C|\Lambda|^{1/d},$ where the constant $C$ depends only on the dimension.

 When $ d \ge 1$ and  $ s \in [\frac 12, 1)$,   \ $ d_{\partial \La}(x)^{-2s} $ is not integrable anymore over $\La$. So we   split the integral  as follows:
 \begin{equation}  \label {t2}  \begin {split}& \int_{\Lambda} \rm {d }x \int_ {\Lambda^c} \rm {d }y\frac { | v(x) - v(y)|^2} {|x-y|^{d+2s}} \cr & =
  \int_{ \{x \in \Lambda:  d_{\partial \La}(x) \le  1 \}} \rm {d }x \int_ { y \in\Lambda^c} \rm {d }y\frac { | v(x) - v(y)|^2} {|x-y|^{d+2s}}   +  \int_{\{x \in \Lambda:  d_{\partial \La}(x)> 1 \}} \rm {d }x \int_ { y \in\Lambda^c} \rm {d }y\frac { | v(x) - v(y)|^2} {|x-y|^{d+2s}}.
 \end {split} \end {equation}
  For the last integral, since $|x-y| \ge 1$,   we  obtain  proceeding  as in \eqref{oct2} 
 \begin{equation}   \label {ele1} \begin {split} & \int_{\{x \in \Lambda:  d_{\partial \La}(x)> 1 \}} \rm {d }x \int_ { y \in\Lambda^c} \rm {d }y\frac { | v(x) - v(y)|^2} {|x-y|^{d+2s}} 
 \cr&  \le  C \int_{\{x \in \Lambda:  d_{\partial \La}(x)> 1 \}} d_{\partial \La}(x)^{-2s}   \rm {d }x     \le   \left \{  \begin {split} &   C |\La|^{\frac {d-1} d}   \qquad s \in ( \frac 12,1) \cr &  C   |\La|^{\frac {d-1} d}  \log |\La|, \qquad s = \frac 12.  
     \end {split} \right.     \end {split} \end {equation}
We split  the first integral   of \eqref {t2}  as 
 \begin{equation}  \label {oct3} \begin {split}&   \int_{ \{x \in \Lambda:  d_{\partial \La}(x) \le  1 \}} \rm {d }x \int_ { y \in\Lambda^c} \rm {d }y\frac { | v(x) - v(y)|^2} {|x-y|^{d+2s}}  \cr &  =
 \int_{ \{x \in \Lambda:  d_{\partial \La}(x) \le  1 \}} \rm {d }x \int_ {\{y \in \Lambda^c:  d_{\partial \La}(y) \le  1  \}} \rm {d }y\frac { | v(x) - v(y)|^2} {|x-y|^{d+2s}}   \cr &+
  \int_{ \{x \in \Lambda:  d_{\partial \La}(x) \le  1 \}} \rm {d }x \int_ { \{y \in \Lambda^c:  d_{\partial \La}(y) > 1  \}} \rm {d }y\frac { | v(x) - v(y)|^2} {|x-y|^{d+2s}}.
    \end {split} \end {equation}
    For  the last  term of \eqref   {oct3}, since  $|x-y| \ge 1  $, we get 
    $$  \int_{ \{x \in \Lambda:  d_{\partial \La}(x)  \le  1 \}} \rm {d }x \int_ { \{y \in \Lambda^c:  d_{\partial \La}(y) > 1  \}} \rm {d }y\frac { | v(x) - v(y)|^2} {|x-y|^{d+2s}}  \le  C \int_{ \{x \in \Lambda:  d_{\partial \La}(x)  \le  1 \}} \rm {d }x\int_1^\infty r^{-1-2s}{\rm d}r\le 
    C |\La|^{\frac {d-1} d}   \qquad s \in [ \frac 12,1).
$$
       The first term of the   right hand side of  \eqref  {oct3} is obviously bounded  when $v \in H^s_{loc} (\R^d)$
   $$   \int_{  \{x \in \Lambda:  d_{\partial \La}(x) \le  1 \} } \rm {d}x \int_ {\{y \in \Lambda^c:  d_{\partial \La}(y)     \le    1 \} } \frac { | v(x) - v(y)|^2} {|x-y|^{d+2s}}  \rm {d}y \le  \|v\|_{H^s(B_{1} (\partial \Lambda))}. $$   
       When   $ v \in C^{0, \alpha} ( B_{1} (\partial \Lambda)) $  for $\alpha > s -\frac 12 $ then again  \ arguing   as in \eqref{oct2} 
  \begin{equation}   \int_{  \{x \in \Lambda:  d_{\partial \La}(x) \le  1 \} } \rm {d}x \int_ {\{y \in \Lambda^c:  d_{\partial \La}(y)   \le    1 \} } \rm {d}y  \frac { | v(x) - v(y)|^2} {|x-y|^{d+2s}}  \le     C |\La|^{\frac {d-1} d},  \qquad s \in [\frac 12,1).    \end {equation}
   
\end{proof} 
   
   \vskip 0.5cm 
     Next we prove an energy decreasing rearrangement which  allows to show  a strong maximum principle,  see Lemma  \ref {FGK}:  Minimizers  of $ G_1 (\cdot,\omega,\La)$ corresponding to ordered boundary conditions  on  $ \Lambda^c$  are ordered as well, i.e  they do not intersect.     In particular if there exists  more than one minimizer   corresponding to the same boundary condition  they do not intersect. 
    \begin{lem} 
  \label{FGK0} Let $u$    and $v$  be in $ H^s_{loc} (\R^d)  \cap L^\infty(\R^d) $. Then  for all $\om \in \Om$ and $\La \subset \R^d$
   \begin {equation} \label {f10}
G_1(u\vee v, \om, \La)+G_1(u\wedge v, \om, \La)\le G_1(u, \om, \La)+G_1(v, \om, \La). 
\end {equation} 
When $u=v$ on $\La^c$,  the equality holds in \eqref {f10}   if and only if
    \begin {equation} \label {f11}u(x) \le v(x)  \quad \hbox {or}  \quad  v(x) \le u(x), \quad  a.s. \quad  x \in  \La.\end {equation} 
    When $u\le v$ on $\La^c$ and $u<v$ for some open set in  $\La^c$ 
    the equality holds in \eqref {f10}   if and only if
    \begin {equation} \label {f12a}u(x) \le v(x)    \quad  a.s. \quad x \in  \La.\end {equation} 

   \end {lem}
   \begin {proof}   Since   $u$ and $v$ are    in $ H^s_{loc} (\R^d)  \cap L^\infty (\R^d) ,$  $G_1$ is finite.          Let   $M(x)= \max \{ u(x), v (x) \}$  and  $ m(x)= \min \{ u(x), v (x) \}$.   It is immediate to verify that     the local part of the functional $G_1$ satisfies    \eqref {f10} with the equality.
For the interaction term, 
  for  $ x$ and $y$ in $\R^d$, we have that
    \begin{equation}  \label {fa1}  [m(x)- m(y)]^2+  [M(x)- M(y)]^2 \le [u(x)- u(y)]^2 + [v(x)- v(y)]^2. 
     \end {equation}
   Namely if  both the minimum values in $x$ and $y$ are reached  by the same function either $u$ or $v$  then the equality holds in  \eqref {fa1}.
   If  $m(x)=u(x) <v(x)$ and $m(y)=v(y)< u(y)$  then  the left hand side of \eqref {fa1}
  is equal to 
   $$ [u(x)- u(y)]^2 + [v(x)- v(y)]^2 +  [u(x)- v(x)]  [u(y)- v(y)] $$ 
   with $ [u(x)- v(x)]  [u(y)- v(y)] <  0$. 
     The same holds when 
    $m(x)=v(x)$ and $m(y)=u(y)$. In these last case we will have a strict inequality in \eqref   {fa1}, and therefore in  \eqref {f10}.
    
  Next we prove  \eqref {f11}.   If $u(x) \le v(x)$ or $u(x) \ge v(x)$ for all $x \in \R^d$ then the equality holds in  \eqref {f10}.  When $u=v$ on $\La^c$ 
       we have also the reverse  implication for $x \in \La$.  Namely it is immediate to verify that in such case (no matter   which value of $u$ or $v$ correspond to $M$ or $m$) 
     \begin{equation}  \label {fa1a}  \WW(M,\La) +   \WW(m,\La)=  \WW(u,\La) +   \WW(v,\La).  \end {equation}
The      equality     in \eqref {f10}  implies the equality       in   \eqref {fa1},     then  \eqref {f11}   holds. 
    Next we prove  \eqref {f12a}.   It is immediate to verify that if     \eqref {fa1a}  holds   we must have  
    $M(x)= v(x)$ and $m(x)=u(x)$ for $x \in \La$.  
        \end {proof}
  \begin{lem} 
  \label{FGK}   Let $u$    and $v$  in $ H^s_{loc} (\R^d)  \cap L^\infty(\R^d) $    be minimizers of $ G_1 $ in $ \Lambda$, so that   
$u \le v $  on $ \Lambda^c $. 
Then, for all $\om \in \Om$,  $u=v$ or  $|u(x)-v(x)|>0$ for all $x\in {\rm int}(\Lambda).$ 
If $u<v$ in an open set in $  \Lambda^c,$ 
then $u<v$ everywhere in ${\rm int}(\Lambda).$
\end{lem}
\begin {proof}  
 Since the  result  holds for any realization of the random field and $\La$ is fixed
 we avoid to explicitly  write  in $G_1$ the dependence on $\om$ and $\La$. 
By Lemma \ref {FGK0} 
\begin {equation} \label {f1}
G_1(u\vee v)+G_1(u\wedge v) \le G_1(u)+G_1(v). 
\end {equation} 
The conditions on $u$ and $v$ in $\Lambda^c$  yield 
   $u\vee v=v, u\wedge v=u$ on $\Lambda^c,$ 
and  by the minimization properties of $u$ and $v$ we get 
$G_1(u\vee v)\ge G_1(v),\ G_1(u\wedge v)\ge G_1(u).$   This     and  \eqref {f1} 
imply that $G_1(u\vee v)+G_1(u\wedge v)=G_1(u)+G_1(v) $,
actually  that  $G_1(u\vee v)= G_1(v)$ and   
$G_1(u\wedge v)= G_1(u)$.
      Therefore 
$u\vee v$ is a minimizer with condition $v$ on $\Lambda^c,$  and 
$u\wedge v$ is a    minimizer with condition $u$ on $\Lambda^c.$
Obviously   the function 
$ w:=  u-u\wedge v\ge 0$  in   $\La$ and  in particular $w=0$ on   $  \Lambda^c$.  Further since $u$ by assumption is a minimizer and  $u\wedge v$  is also a minimizer,   they  are both solutions of problem \eqref {EL.1} and 
the regularity results of Proposition \ref {Lip} hold.

Therefore  by construction    $w \in C^{0,\alpha}(\La)$,    $\alpha <2s$,   when  $2s \le 1$ 
and     $w \in    C^{1,\alpha} (\La) $ ,   $\alpha <2s-1$,   when   $2s > 1$.   On the other hand,     $ w$   solves

\begin{equation}  \label{a11} \begin{split}    
&   (-\Delta)^s w =  V(x)     \quad 
\text{in } \Lambda,  \\ &
 w =0 \qquad   \hbox {on} \qquad   
 \Lambda^c 
 \end{split}
\end{equation}
where 
$$
V(x)=\frac{1}{2} [W'(u(x)) - W' (u(x)\wedge v(x))].
$$
 Since $ W \in C^{2} (\R)$,  see Assumption     (H1),  by the   regularity of $u$ and  $u\wedge v$
 we have that $V \in   C^{0,\alpha}(\La)$,    $ 0<\alpha <2s$,   when  $2s \le 1$ 
 and    $V \in    C^{1,\alpha} (\La) $ ,   $0<\alpha <2s-1$,   when   $2s > 1$. 
  By     \cite [Proposition 2. 8]{LS}    $w$  being solution of \eqref {a11}  is  in $C^{0,\alpha +2s}$  when  $\alpha +2s \le 1$ and in  $C^{1,\alpha +2s-1}$  when  $\alpha +2s > 1$.  In both cases the following argument holds.  
 Suppose there exists $x_0 \in  \Lambda  $  with  $u (x_0)=u (x_0)\wedge v (x_0)$, i.e  $w(x_0)=0.$
 By the regularity  of $w$ we have that     
 $$(-\Delta)^{s}w(x_0)=    \int_{\Lambda} \rm {d }y\frac { [ w(x_0) - w(y)] } {|x_0 -y|^{d+2s}} = - \int_{\Lambda} \rm {d }y\frac {   w(y)  } {|x_0-y|^{d+2s}}<0$$
 being equal to zero only when $w(x)=0$  for  almost all  $x \in  \R^d$.  Notice that the integral is well defined for any $s \in (0,1)$ since $w$ is  is  in $C^{0,\alpha +2s}$  when  $\alpha +2s \le 1$ and in  $C^{1,\alpha +2s-1}$  when  $\alpha +2s > 1$.
Since  $V(x_0)=0$ by construction, if  $(-\Delta)^{s}w(x_0) <0$  we have a contradiction  with  \eqref {a11}. 
  
Therefore in the interior of $\Lambda$ 
either $u=u\wedge v$ (in which case $u\le v$)  or $u>u\wedge v,$ 
i.e. $v<u$.  By assumption $u \le v$ in $\La^c$  and by Lemma \ref {FGK0}
 $v<u$  in the interior of $\Lambda$ is only possible if $u=v$ on $ \Lambda^c.$  
 Next we  show that when $u=u\wedge v$, then either   $u=v$ in  $\La$ (and this is possible only when  $u=v$  on $\La^c$)   or $u(x)  <v(x)$ for $x$ in the interior of $\La$.  Denote by $w= u-v \ge 0$. As before, we have that $w$ is  a solution of  
 \begin{equation}  \label{a11a} \begin{split}    
&   (-\Delta)^s w =  V(x)     \quad 
\text{in } \Lambda,  \\ &
 w =w_0\ge 0  \qquad   \hbox {on} \qquad   
 \Lambda^c,
 \end{split}
\end{equation}
where we set $w_0= v-u$, the difference of the boundary data, which by assumption is positive.
Arguing as before, assume that there exists $x_0$ in the interior of $\La$ so that $w(x_0)=0$. 
By the regularity  of $w$ we have that 
\begin {equation} \label {f12} \begin{split} &(-\Delta)^{s}w(x_0)=    \int_{\Lambda} \rm {d }y\frac { [ w(x_0) - w(y)] } {|x_0 -y|^{d+2s}}  + 
  \int_{\Lambda^c} \rm {d }y\frac { [ w(x_0) - w_0(y)] } {|x_0 -y|^{d+2s}}  
\cr &  = - \int_{\Lambda} \rm {d }y\frac {   w(y)  } {|x_0-y|^{d+2s}}  -
   \int_{\Lambda^c} \rm {d }y\frac { w_0(y)  } {|x_0 -y|^{d+2s}}   <0.
   \end{split}
\end{equation}
 Since  $V(x_0)=0$ by construction, if  $(-\Delta)^{s}w(x_0) <0$  we have a contradiction  with  \eqref {a11a}. 
Therefore  if $w_0=0$ in $\La^c$, in the interior of $\Lambda$ 
either $w=0$ (in which case $u= v$)  or $v>u$.
If $w_0>0$ in some subset of $\La^c$  the only possibility is $v>u$ in the interior of $\La$.  
 \end {proof}
  Note that there may be a priori several minimizers with the same boundary conditions, as our functional is not convex. 
 
Next,  given $ v_0 \in    H^s_{loc}(\R^d) \cap  L^\infty ( \R^d) $,   we single out two special minimizers    of $G_1$ in $\Lambda$, one is the largest minimizer of $G_1$ in $\La$ with $v_0$ boundary conditions (defined a pointwise supremum), the other is the smallest minimizer   of $G_1$ in $\La$ with $-v_0$ boundary conditions.  We call them 
 the  $v_0-$  maximal  and the   $v_0-$ minimal minimizer of $G_1$ in $\La$. 
  
  \begin{lem} { \bf Existence of maximal/minimal minimizers.}  
  
  Let $\La \Subset \R^d$ be a  Lipschitz  bounded open set and $ v_0 \in    H^s_{loc}(\R^d) \cap  L^\infty ( \R^d) $.  
  \begin{enumerate} \item 
 The set of  minimizer of $G_1^{v_0}$  on   $\Lambda$ is compact, i.e.  any sequence of minimizers has a  limit  in $C^{0,\alpha} (\La),$ $\alpha<2s$ for
 $s\in(0,1/2] $ or $C^{1,\alpha} (\La) ,$ $\alpha<2s-1$ for $s\in (1/2,1),$ which is still a minimizer. 
 \item The set of minimizers has a maximal and minimal element with respect to point-wise ordering of functions
 \end{enumerate}
 \end{lem}
{\em Proof:} A sequence of minimizers  of $G_1^{v_0}$  on   $\Lambda$ is  a sequence of functions with energies  converging  to the infimum, so the same techniques as in the proof of the existence of minimizers apply.

For the second part, let us  define a function $\bar u:\Lambda\to \R$ by $\bar u(x):=\sup\{v(x):\ v\  {\rm minimizer}\}.$
We have to show that $\bar u$ is a minimizer, in particular that it has sufficient regularity. Fix a point $x_0$ in the interior of $\Lambda.$ We can find a sequence of minimizers $\{v_n\}_{n\in \N}$  (which for the moment may still depend on $x_0$)  such that 
$v_n(x_0)\to \bar u(x_0)$ and such that the sequence $v_n(x_0)$ is increasing. By Lemma \ref{FGK},
$v_n(x)\le v_m(x)$ for all $m\ge n$ and {\em all} $x\in \Lambda.$ Define now $\bar v(x):=\lim_{n\to\infty}v_n(x).$
We know from  the  first part of the Lemma that the sequence of minimizers $\{v_n\}_{n\in \N}$ has a convergent subsequence 
which converges to a minimizer. So the point-wise limit $\bar v$ must be   minimizer, moreover $\bar v\le \bar u.$\newline
If there exists $x_1\in \Lambda $ such that $\bar v(x_1)<\bar u(x_1),$ then there must be a minimizer
$w$ such that $w(x_1)>\bar v(x_1).$ But $\bar v(x_0)=\bar u(x_0)\ge w(x_0),$ contradicting Lemma \ref{FGK}.
So $\bar v=\bar u,$  which is therefore  the maximal minimizer and point-wise maximum over the set of minimizers. The proof for the minimal element is done in the same way.

 This allows us to define the following object:
 
      \begin {defin}  \label {exmin} Given $ v_0 \in  H^s_{loc}(\R^d) \cap  L^\infty ( \R^d)  $, we  say that $u^+$ ( $u^-$) is the  $v_0-$ maximal ( $v_0-$ minimal) minimizer of $G_1$ in $\Lambda$ if 
 \begin {itemize}
 \item
$u^+(x) =v_0(x)$,   $ (u^-(x) = -v_0(x))$  for $x \in \Lambda^c$,
\item
 $u^+$, ($u^-$) is a miminizer of $G_1$ in $\Lambda $ according  to (2) of Definition \ref {min0}, 
 \item if  $ \tilde u$ is any  other minimizer (if more than one)   of $G_1$ in $\Lambda $ so that   $ \tilde u(x)= v_0(x)$, $ (\tilde u(x)= -v_0(x))$   for $x \in \La^c$,   then
 $ \tilde u(x) <u^+(x)$      $( \tilde u(x) >u^-(x))$  for $x \in \Lambda$. 
 \end  {itemize}
\end {defin}

     \medskip

   \section{Infinite volume covariant states}
   
In this section we  construct two functions $   v^{\pm} (\cdot, \om) $   on  $\R^d$ which we denote {\it macroscopic extrema minimizers  or  infinite-volume states.}   
  They are obtained, as explained in the introduction,   through a two limits procedure.
 We first show that  for any      $K\ge   1 + C_0 \theta \|g\|_\infty$,  where $C_0$  is the constant in 
  \eqref {V.1},  the $K-$ maximal and minimal minimizers of $G_1^{K}$ in $ \Lambda_n$   as $ n\to \infty$
converge  in a   suitable  way    to $u^{\pm,K}$.  Then  we define  the  $v^{\pm} (\cdot, \om) $ as the point-wise  limit, when    $K \to \infty$  of  $u^{\pm,K}$. 
   We show that the $   v^{\pm} (\cdot, \om) $, constructed in such a way,   are minimizers under compact perturbations  
and    they  do not  depend on the    boundary values.  

 \begin{thm} \label{infvol}  [infinite-volume states]
For almost all   $\om \in \Omega$,    there exist  two functions $v^+(x,\omega),$ $v^-(x,\omega)$, $ x \in \R^d$,  having the following properties.
 \begin{itemize}
 \item   If $2s\le 1$, then $v^{\pm} (\cdot, \om) \in C^{\alpha}_{loc} (\R^d)$ for all $\alpha <2s$. 
If  $s \in (\frac 12, 1),$ then   $v^{\pm} (\cdot, \om)  \in C^{1,\alpha}_{loc}(\R^d)$ for all $ \alpha <2s-1$.
   \item    $v^{\pm} (\cdot, \om) $ are  translation covariant.
   \item  \begin {equation}  \label {c1} v^+ (x, \om) = -v^- (x, -\om) \quad x \in \R^d.  \end{equation}
  \item $v^\pm$ are minimizers under compact perturbations in the sense of Def. \ref{min0}, (1). 
   \item     \begin {equation}  \label {eq:bound}    \| v^{\pm}   (\om)\|_\infty  \le 1+ C_0 \theta \|g\|_\infty,  \end{equation}    where $C_0$ is the constant in 
  \eqref {V.1}. 
 \item Let $\Lambda_n = (-\frac n 2, \frac n 2)^d$, $n \in \N$, we have 
\begin {equation} \label {M1a} 
\lim n^{-d}\int_{\Lambda_n}v^\pm (x, \omega) {\rm d} x =m^\pm,
\end {equation} where $m^\pm =  \E \left [  \int_{   [-\frac 12, \frac 12]^d}   v^\pm(x, \cdot )  {\rm d} x \right ] $, and $m^+=-m^- \ge 0$. 
 
  \item  Given  $ v_0 \in L^\infty (\R^d)$, let  $\bar w_n (\cdot, \om)$   be  a minimizer of $  
 G_1^{v_0} (v, \om,\La_n)$ according to Definition \ref {min0},  then  uniformly on $v_0$   
 \begin{equation}  \label {diseq1}  v^- (x, \omega) \le  
\liminf_{n\to\infty} \bar w_n (x , \om) \le\limsup_{n\to\infty} \bar w_n (x , \om) 
\le v^+ (x, \omega),  \quad \hbox {(uniformly on compacts in $x$)}. \end{equation}
\end{itemize}
 \end{thm}
 These $v^{\pm} (\cdot, \om) $ infinite volume minimizers will be obtained as limits of the so-called $K$-maximal/minimal minimizers.
 
 \begin{prop}\label {ag1}
Let $K\in \R,\ K\ge 1+ C_0 \theta \|g\|_\infty$   and $u^{\pm ,K}_n \in   H^s_{loc} (\R^d) \cap L^\infty(\R^d)    $ be  respectively the $K-$maximal and the $K-$minimal  minimizers  of $G_1$ in $ \Lambda_n= (-\frac n 2, \frac n 2)^d$.   We have that  $\Pr-$a.s.
        \begin {equation}  \label {ag8}  \lim_{n \to \infty}  u^{\pm,K}_n (x, \om)= u^{\pm,K} (x, \om) \qquad \hbox  {  point-wise and uniformly on compacts  in } \quad  x.  \end {equation}
  Further  \begin{itemize}
 \item   
  If $2s\le 1$, then $u^{\pm,K} (\cdot, \om) \in C^{\alpha}_{loc} (\R^d)$ for all $\alpha <2s$. 
If  $s \in (\frac 12, 1),$ then   $u^{\pm,K} (\cdot, \om)  \in C^{1,\alpha}_{loc}(\R^d)$ for all $ \alpha <2s-1$.
  \item    $u^{\pm,K} (\cdot, \om) $ are  translation covariant.
  \item   \begin {equation} \label {f25} u^{+,K} (\cdot, \om)=  - u^{-,K} (\cdot, -\om), \quad \Pr- a.s.\end {equation}. 
 \end {itemize}
\end{prop}
  \begin {rem}  \label {ra3}   As in Remark \ref {ra2}   the convergence  in \eqref {ag8} holds  in $C^{0,\beta}$,  $ \beta <\alpha < 2s$  when $s \in (0, \frac 12 ]$ and in  $C^{1,\beta}$,  $ \beta < \alpha$, $  \alpha <2s-1$ when $s \in (\frac 12,1)$.
\end  {rem}  
 \begin{proof} 
We start proving the existence of $u^{\pm,K}$.
  For $z \in \Z^d$,     denote by  
   $ u_n^{z,+,K} := u_n^{z,+,K}(\cdot, \om) $  the maximal minimizer of $G_1$ in the domain $ z + \La_n$,  so that  $u_n^{z,+,K} (\cdot, \om)=  K$ in $ \R^d \setminus (z+\La_n)$ 
   and respectively   $ u_n^{z,-,K}:= u_n^{z,-,K}(\cdot, \om) $  the minimal minimizer of $G_1$ in the domain $ z + \La_n$ ,  so that  $u_n^{z,-,K} (\cdot, \om)= -K$ in $ \R^d \setminus (z+\La_n)$.  If $z=0$ we write $u_n^{\pm,K}.$     Without loss of generality  we   assume for the next paragraph $z=0.$ 
   
  By   Lemma \ref{A}, (2),   $\|u^{\pm,K }_n\|_\infty\le K$ for any $n$. 
Therefore   $u^{+,K }_m \le  K$
on $\Lambda_m \setminus  \Lambda_n$ for $m>n$.  Lemma \ref{FGK} implies
that for any $x$ and $\omega$   (and $n>n_0(x)$ ) the sequence
$\{u^{+,K}_n(x,\om)\}_n$  is decreasing.      Moreover it is bounded from below
by $-K$.   Hence, reasoning in a similar manner for $u^{-,K }_n,$
$$
u^{\pm,K }(x,\omega):=\lim_n u^{\pm,K }_n(x,\omega)
$$ exist and are  measurable as function of $\om$.  
 We start analyzing  the case $2s\le 1$.  As the $u^{\pm,K }_n$ are bounded and
minimizers,  they are on each fixed cube $Q$  H\"older continuous of order $\alpha <2s$ for any $2s\le1,$ uniformly in $n,$ provided $Q\subseteq\Lambda_n,$ see  Proposition \ref {Lip}.   This implies that subsequences converge locally uniformly
to a H\"older function  of order $\alpha <2s$. As the entire sequence converges point-wise,
the limit of any subsequence must coincide with $u^{\pm,K},$ which is therefore a locally  H\"older continuous function  of order $\alpha <2s$. The same argument for general $z$ yields
monotone limits $u^{z,\pm,K}.$ 
When $s \in (\frac 12, 1)$  the argument goes in the same way, the only difference is that   
 by   Proposition \ref {Lip}  
  the  minimizers $u^{\pm,K}_n$ are  
 uniformly bounded and uniformly   with respect to $n$ in $C^{1,\alpha}$    with $\alpha <2s-1$  on each fixed cube $Q$ which does not depend on $n$.  
 
 To show   that $u^{\pm,K}$ are translational covariant,    
notice that, by \eqref{parisv1}
$$
u_n^{0,+,K}(0,\omega)=u^{z,+,K}_n(z, T_{-z}\omega).
$$  
Take $m$ large enough so that    $\Lambda_n+z\subseteq\Lambda_{m}$.
  We  have that  $u^{z,+,K}_n(z, T_{-z} \om) =  u_n^{0,+,K}(0,\omega) \ge u^{0,+,K}_{m}(0, \om)$, since $m>n$.
  Then letting first  $n \to \infty$ and then $m\to\infty$ we get   $u^{z,+,K} (z, T_{-z} \om) \ge u^{0,+,K} (0, \om) $.
  The opposite equality follows in the same way by taking  $ \Lambda_m \subseteq \Lambda_n+z$. 
Note that we used in the proof that the boundary condition  is translation invariant.
 Next we prove \eqref  {f25}. 
 It is immediate to verify that
  \begin{equation} \label{c2} G_1^{K} ( v, \om, \Lambda_n) =     G_1^{-K}(-v, -\om,\Lambda_n)=  G_1^{-K}(w, -\om,\Lambda_n),
\end{equation}
 (see  notation  \eqref  {funct110}), if  we set $-v = w$. 
 Therefore if $   u^{+,K}_n(\cdot, \omega)$ is the  maximal minimizer of  $ G_1^{K} ( v, \om, \Lambda_n) $   we have that   $   w _n (\cdot, -\omega) = - u^{+,K}_n(\cdot, \omega)$ is 
the minimal minimizer of  $ G_1^{-K}(w, -\om,\Lambda_n)$   in $ \Lambda_n$, i.e 
$    w _n (\cdot, -\omega)  =  u^{-,K}_n (\cdot, -\omega)$.   Then letting $n\to \infty$ we get \eqref {f25}. 
\end{proof}  

Next we show  that the states $u^{\pm,K}$     are indeed minimizers under compact perturbations.
In the proof we will only use that the boundary condition is bounded by $K$ and has the regularity of a minimizer, but not that it is actually a constant. 

 \begin{prop}  \label{jt2}   Let       $K \in \R$,  $K\ge   1 + C_0 \theta \|g\|_\infty$  and $u^{\pm,K} (\cdot, \om)$ be   the functions constructed  in  Proposition   \ref {ag1}.
 Then,   for any   $\La \Subset \R^d$,  we have that
 $$ G_1^{u^{+,K}}(u^{+,K}, \om, \La) \le G_1^{u^{+,K}}(u, \om, \La),$$
  for any  measurable function $u$ which coincides with $u^{+,K} (\cdot, \om)$ in $\La^c$. 
  The same holds for $u^{-,K}$.  
\end {prop}
   
  \begin{proof}
 Denote shortly $u^{+,K}= u^*$. 
 We argue by contradiction. Assume that there exists a bounded     set $ \Lambda$  and a measurable  function $u$ so that
 $ G_1^{u^*}(u, \om, \La)<  G_1^{u^*}(u^*, \om, \La)$.    
       Let  $ \Lambda_n$ be so large that  $ \Lambda \subset \Lambda_n$  and let  $u_n^{+,K}$  be  the $K-$maximal minimizer of $G_1$ in $\Lambda_n$,  see  Definition \ref {exmin}.

For simplicity we drop the dependence on $\omega $ and  denote 
  $$E_1:=G_1^{u^*}(u^*,\Lambda), \qquad E_2:=G_1^{u^*}(u,\Lambda), \qquad E_n:=G_1^{K}(u^{+,K}_n,\Lambda_n). $$ 
By  assumption
  there exists a $\delta>0$ such that $E_2+\delta<E_1.$ The aim is to construct a function $\widetilde{u_n}$ such that  if  $E_2+\delta<E_1 $ then 
  $G_1^{K}(\widetilde{u_n},\Lambda_n)<E_n$ for some   $n$ large enough,  which  gives   a contradiction. 
  \vskip0.5cm 
  \noindent 
  \emph{ Step 1}:  By \eqref {funct11}
  \begin{eqnarray}\label{split1jt2} E_1\!\!&=&\!\!\KK_1(u^*,\Lambda)+\WW((u^*,\Lambda),(u^*,\Lambda^c))   \\
  \label{split2jt2} E_2\!\!&=&\!\!\KK_1(u,\Lambda)+\WW ((u,\Lambda),(u^*,\Lambda^c))  \\
  \label{split3jt2} E_n\!\!&=&\!\!\KK_1(u^{+,K}_n,\Lambda)+\WW((u^{+,K}_n,\Lambda),(u^{+,K}_n,\Lambda_n\setminus\Lambda))+
  \KK_1(u^{+,K}_n,\Lambda_n\setminus\Lambda) \cr &+&\WW((u^{+,K}_n,\Lambda_n\setminus\Lambda),(K,\Lambda_n^c)).  
  \end{eqnarray}

 \vskip0.5cm 
  \noindent
  \emph{ Step 2:}  Next we show that for any $\eps>0$ there exists  $n_\eps$ s.t. for $ n \ge n_\e$
   \begin{eqnarray}
  \label{est1bjt2} |\KK_1(u_n^{+,K},\Lambda)-\KK_1(u^*,\Lambda)|<\eps,\\
  \label{est2bjt2}
A\equiv  | \WW ((u^*,\Lambda),(u^*,\Lambda_n\setminus\Lambda))- \WW ((u_n^{+,K},\Lambda),(u_n^{+,K},\Lambda_n\setminus\Lambda))|<\eps.
  \end{eqnarray}
  The  bound  \eqref {est1bjt2}  follows immediately from    Proposition \ref{convest} with $D=\Lambda$,   the regularity properties of the minimizers and Remark \ref {ra3}.  
      To show  \eqref {est2bjt2},  fix   $R>0$ so that  $\Lambda \subset B_{R/2}(0)$  and require $n$ so large that $B_R(0)\subset \Lambda_n.$  Note that we can choose such $R$ to be bounded uniformly in $n.$ 
  We  upper bound $A$ in \eqref {est2bjt2} as following:
  \begin{eqnarray*}
 &&  A \le |I_1|+|I_2|, \\
 &&I_1=\int_{\Lambda}\int_{B_R(0) \setminus\Lambda}\frac{|u^*(z)-u^*(z')|^2-|u_n^{+,K}(z)-u_n^{+,K}(z')|^2}{|z-z'|^{d+2s}}dz dz',\\
 &&I_2=\int_{\Lambda}\int_{\Lambda_n\setminus B_R(0)}\frac{|u^*(z)-u^*(z')|^2-|u_n^{+,K}(z)-u_n^{+,K}(z')|^2}{|z-z'|^{d+2s}}dz dz'.
 \end{eqnarray*}
 $I_1$ is estimated (in a very rough way) by Proposition  \ref{convest} with $D=B_R$. 
For $I_2$,   since  $|u^*|\le K,\ |u_n^{+,K}|\le K$   
  we  have
  $$
  |I_2|\le\int_{\Lambda}\int_{\R^d\setminus B_R(0)}\frac{8K}{|z-z'|^{d+2s}}dz dz'\le 8KC(d)|\Lambda|\int_{R/2}^\infty r^{-2s-1}\le K
  |\Lambda|C'(d)R^{-2s}.
  $$
Here we used the  integrability of the kernel at infinity. In conclusion, by choosing first $R$ sufficiently large, depending on $\eps,$ and then $n_\eps$ large
depending on $R$ we obtain  \eqref{est1bjt2} and \eqref{est2bjt2}  for all $n\ge n_\eps.$
 \vskip0.5cm 
\noindent
\emph{ Step 3:} In the same way as $I_2$ above we use the integrability of the kernel at infinity to get  
  $$
  | \WW ((u^*,\Lambda),(u^*,\Lambda_n\setminus\Lambda))-\WW_1((u^*,\Lambda),(u^*,\R^d\setminus\Lambda))|
  \le 4KC(d)|\Lambda|\int_{R/2}^\infty r^{-2s-1}\le K
  |\Lambda|C'(d)R^{-2s}
<\eps
  $$for $R$ and $n$ sufficiently large. So
   \begin {equation} \label {ag6} \begin{split}
  E_n >&E_1-3\eps  +\KK_1(u_n^{+,K},\Lambda_n\setminus\Lambda)+\WW ((u_n^{+,K},\Lambda_n\setminus\Lambda),(K,\Lambda_n^c)) \\
   >&E_2 +\KK_1(u_n^{+,K},\Lambda_n\setminus\Lambda)+\WW_1((u_n^{+,K},\Lambda_n\setminus\Lambda),(K,\Lambda_n^c))+\delta -3\eps . 
  \end{split}\end {equation}
  
  \noindent 
  \emph{ Step 4}  Now we  construct  a function on $\Lambda_n$ such that its energy in this cube with $K$ b.c. approximates the first three  terms in the
 last line
of \eqref{ag6}, which will lead to a contradiction. 
  Define   a function  $\widetilde{u_n}$ which is equal to $u$ in $\Lambda$ and equal to $u_n^{+,K} $ outside a boundary layer of width 1 of    $\Lambda $: 
   \begin {equation} \label {ag2}\widetilde{u_n}(x):=\left\{\begin{array}{ll}u(x),&\ {\rm if\ } x\in \Lambda,\\ u_n^{+,K}(x)&  {\rm if\ }    x\in\R^d:\ {\rm dist}(x,\Lambda)>1, \\  u^*(x)+\Psi(x)(u_n^{+,K}(x)-u^*(x)) & {  \rm else} 
  \end{array}\right.
  \end {equation} 
   where 
  $\Psi:\R^d\to [0,1]$ is a   smooth cut-off function nondecreasing in ${\rm dist}(x, \Lambda)$  with $\Psi(x)=0$ if ${\rm dist}(x,\Lambda)<1/2$ and 
$\Psi(x)=1$ if ${\rm dist}(x,\Lambda)>1.$    Notice that    $\widetilde {u_n} -   u^* \to 0$ in $C^{0,\alpha} (\Lambda_n\setminus\Lambda)$ for $\alpha<2s.$ 
 By  the equality    $ u^*(x)+\Psi(x)(u_n^{+,K}(x)-u^*(x))   =u_n^{+,K}(x)+[1-\Psi(x)](u^*(x)-u_n^{+,K}(x)) $ which  we will use   in the following we  get  also that 
 $\widetilde {u_n} -   u_n^{+,K} \to 0$ in $C^{0,\alpha} (\Lambda_n\setminus\Lambda)$ for $\alpha<2s.$
  Set
 \begin{eqnarray*}
 &&I_3=|\WW ((u,\Lambda),(u^*,\Lambda^c)) -\WW ((u,\Lambda),(\widetilde{u_n},\Lambda^c))|\\&&=\left|\int_{\Lambda } dz \int_{\Lambda^c} dz'\frac{|u(z)-u^*(z')|^2-|u(z)-\tilde u_n (z') | ^2}{|z-z'|^{d+2s}}\right|\\&& =\left |\int_{\Lambda } dz \int_{\Lambda^c} dz'\frac{2    u(z) [ \widetilde{u_n}(z')-u^*(z')] +   [ (\widetilde{u_n}(z'))^2 -(u^*(z'))^2]   
}{|z-z'|^{d+2s}}\right |.
 \end{eqnarray*}
As $\widetilde {u_n}(x)=u^*(x)$ for  $x \in \La^c$ and  ${\rm dist}(x,\Lambda)<1/2,$ the integrand vanishes unless $|z-z'|>1/2.$  For $R$ as in Step 2  we estimate  $I_3$ by splitting $\Lambda^c=(\Lambda^c\cap B_R(0))\cup (\Lambda^c\setminus B_R(0))$
 $$
 I_3\le C(d)|\Lambda|R^d\|u^*-u^{+,K}_n\|_{L^\infty(B_R)}+|\Lambda|C(d)R^{-2s} K.
 $$Choosing first $R$ large and then $n_0$ depending on $R$ and $\epsilon,$ we obtain  that for $n \ge n_0$,  $|I_3|<\epsilon$ and hence, see \eqref {split2jt2}, 
  \begin {equation} \label {ag3} E_2\ge {\mathcal K}_1(u, \Lambda)+\WW ((u,\Lambda),(\widetilde{u_n},\Lambda^c))-\epsilon.  
 \end {equation} 
By definition of $\widetilde{u_n}$
$$\WW ((u,\Lambda),(\widetilde{u_n},\Lambda^c))= 
  \WW ((u,\Lambda),(\widetilde{u_n},\Lambda_n\setminus\Lambda))+\WW_1
 (u,\Lambda),(K,\Lambda_n^c), $$  we therefore obtain
  \begin {equation} \label {ag5}
 E_2\ge {\mathcal K}_1(u, \La)+\WW_1((u,\Lambda),(\widetilde{u_n},\Lambda_n\setminus\Lambda))+\WW
 ((u,\Lambda),(K,\Lambda_n^c))- \epsilon.
\end  {equation}
 
 \noindent
  \emph{ Step 5} 
  By  \eqref {ag2} and \eqref {ag5}
\begin{eqnarray*}
G_1^{K}(\widetilde{u_n},\Lambda_n)&=&
{\mathcal K}_1( u, \Lambda )+\WW ((u,\Lambda),(\widetilde{u_n},\Lambda_n\setminus\Lambda))+\WW(
 (u,\Lambda),(K,\Lambda_n^c))\\&& +{\mathcal K}_1(\widetilde{u_n}, \Lambda_n\setminus\Lambda)+\WW(
 (\widetilde{u_n},\Lambda_n\setminus\Lambda),(K,\Lambda_n^c))\\&\le& E_2+ \eps+
 {\mathcal K}_1(\widetilde{u_n}, \Lambda_n\setminus\Lambda)+\WW(
 (\widetilde{u_n},\Lambda_n\setminus\Lambda),(K,\Lambda_n^c)). 
 \end{eqnarray*} 
 Therefore
 \begin {equation} \label {g1}     E_2 \ge   G_1^{K}(\widetilde{u_n},\Lambda_n) - \eps   -  {\mathcal K}_1(\widetilde{u_n}, \Lambda_n\setminus\Lambda)- \WW(
 (\widetilde{u_n},\Lambda_n\setminus\Lambda),(K,\Lambda_n^c)).  \end  {equation}
By  \eqref {ag6}  if we show 
 that
\begin{eqnarray}\label{est3bjt2}
&&\left| {\mathcal K}_1(\widetilde{u_n}, \Lambda_n\setminus\Lambda) - \KK_1(u_n^{+,K},\Lambda_n\setminus\Lambda)\right|<\eps\\ 
\label{est4bjt2}
&&\left| \WW(
 (\widetilde{u_n},\Lambda_n\setminus\Lambda),(K,\Lambda_n^c)) - \WW ((u_n^{+,K},\Lambda_n\setminus\Lambda),(K,\Lambda_n^c))\right|<\eps,
\end{eqnarray}
then
$$
E_n>-6\epsilon+\delta+G_1^{K}(\widetilde{u_n},\Lambda_n)
$$for $n$ sufficiently large. As $\eps$ was arbitrary and $E_n$ is  minimal  value with $K$-boundary conditions, this means $\delta=0$
and hence $u^*$ is a minimizer under compact perturbations.
Next we prove  \eqref{est3bjt2} and  \eqref{est4bjt2}.  
The \eqref{est3bjt2} follows by applying  Proposition  \ref{convest}   since 
$\widetilde {u_n} -   u_n^{+,K} \to 0$ in $C^{0,\alpha} (\Lambda_n\setminus\Lambda)$ for $\alpha<2s.$    Note that the difference is equal to zero for
${\rm dist}(x,\Lambda)>1$.   
The  \eqref{est4bjt2}  follows by
\begin{eqnarray*}
&&\int_{(\Lambda_n\setminus\Lambda)  \times\Lambda_n^c}  dz dz'\frac{\big|| \widetilde{ u_n} (z)-K|^2-|u_n^{+,K}(z)-K|^2\big|}{|z-z'|^{d+2s}}\le  4 K
\int\limits_{\{{\rm dist}(x,\Lambda)\le 1\}\cap(\Lambda_n\setminus\Lambda) \times B_R(0)^c}  dz dz' \frac{|\widetilde{ u_n} (z)-u_n^{+,K}(z)|}{|z-z'|^{d+2s}}\\ &&\le
|2\Lambda| 4K\|u_n^{+,K}-u^*\|_{L^\infty(2\Lambda)}C(d)\int_R^\infty r^{-2s-1}dr\le
|2\Lambda| 4K \| u_n^{+,K}-u^* \|_{L^\infty(2\Lambda)}C'(d)R^{-2s}
\end{eqnarray*}
where we used that $\widetilde{u_n}=u_n^{+,K} $ for ${\rm dist}(x,\Lambda)>1$ and   $R$ is chosen as large as possible with $B_R(0)\subseteq \Lambda_n.$

    \end{proof} 
  
  Next we   show   that $\|u^{+,K} \|_\infty $   is  bounded uniformly   on $K$.
 
 \begin{lem} \label {Ma1}
Let  $u^{\pm,K }(x,\omega)$ the functions  constructed in Proposition \ref {ag1}, see  \eqref  {ag8}.
 Then 
 {\em uniformly in $K$ }   
  \begin {equation}  \label {eq:bound1}  \|  u^{\pm,K}  (\om)\|_\infty  \le 1+ C_0 \theta \|g\|_\infty  \qquad \Pr- a.s.,\end{equation}    where $C_0$ is the constant in 
  \eqref {V.1}.
\end{lem} 
\begin {proof}
 Take   $\Lambda_0= [-\frac 12, \frac 12]^d$    and     $\La_n= (-\frac n2, \frac n2)^d,$  i.e. $|\La_n|=n^d|\La_0|$.   Define for $z\in \Z^d$ and for any $C^+ \ge 1+ C_0 \theta \|g\|_\infty  $
$$
\Lambda_z:=\Lambda_0+z,\quad  B^n(\omega)=  |\{x\in\Lambda_n : |u^{+,K }(x, \omega)|>C^+\}| = \sum_{ z \in \La_n \cap \Z^d}   |\{x\in\Lambda_z : |u^{+,K }(x, \omega)|>C^+\}|.  $$ 
Since 
$u^{+,K}$  is  translation  covariant
$$|\{x\in\Lambda_z : |u^{+,K }(x, \omega)|>C^+\}|  =  |\{x\in\Lambda_0 : |u^{+,K }(x, T_z \omega)|>C^+\}|.  $$
Hence we obtain 
$$ B^n(\omega)=     \sum_{ z \in \La_n \cap \Z^d}  |\{x\in\Lambda_0 : |u^{+,K }(x, T_z \omega)|>C^+\}|. $$
If  $ |\{x\in\Lambda_0 : |u^{+,K }(x, \omega)|>C^+\}| =0,$  $ \Pr-$ almost surely then  $B^n(\omega)=0 $,    $ \Pr-$ a. s.     for all $n$, and  we obtain the claim.
 Suppose that the claim is false.  Assume  that    for some
$\eta>0$ $$\E \left \{ |\{x\in\Lambda_0 : |u^{+,K }(x, \omega)| \right \}=|\Lambda_0|\eta.$$  Therefore by the ergodic theorem 
$
 \frac { B^n(\omega)} {n^d}  \to \eta   
$ almost surely.  Fix an $\omega$ in the set of full measure where this holds, and treat it from now on as parameter. There exists $n_0 $ (depending on
$\omega$), such that for $n\ge n_0$,  $B^n(\omega)>  n^d \eta/2  >0.$ 

Now   define a function 
  \begin {equation} \label {ag9}v(x):=\left\{\begin{array}{ll}C^+ \wedge u^{+,K} \vee (- C^+),&\ {\rm if\ } \{x:\ {\rm dist}(x,\Lambda_n^c)>2\} ,\\ u^{+,K}&  {\rm if\ }    x\in\R^d:\  \{x:\ {\rm dist}(x,\Lambda_n^c)\le 1\}  \cup  \Lambda_n^c \\  \Phi(x) &{\rm else,}
  \end{array}\right.
  \end {equation} 
where $\Phi(x)$ is a smooth interpolation between     $u^{+,K}$ and   $C^+ \wedge u^{+,K} \vee (- C^+).$  
By Lemma \ref{A} there exists a constant $c >0$ { which 
 depends on   $K$, $\theta$, $C_0$ and $\|g\|_\infty$ } such that
 \begin {equation} \label {ag10}  \KK_1(u^{+,K},\Lambda_n) \ge \KK_1(v,\Lambda_n)+c|B^n| > \KK_1(v,\Lambda_n)+ c  \frac \eta 2 n^d . \end {equation} 
Note that $u^{+,K}$ is  minimizer and has therefore higher regularity. The cutting and interpolation procedure retains Holder regularity. (For the cutting, note that it is the application of a Lipschitz function. For the interpolation, note that the cut-off can be chosen smooth, with a uniform bound on the first derivative.) So we have sufficient regularity to apply Lemma \ref{tec1} for $\Lambda_n,$ and  we know that for any given $\epsilon$  there exists  $n_\e$ sufficiently large so that for $n \ge n_\e$ 
 \begin {equation} \label {ag11}\WW((u^{+,K},\Lambda_n), (u^{+,K},\Lambda_n^c))-\WW((v,\Lambda_n), (u^{+,K},\Lambda_n^c))\ge -2\eps|\Lambda_n|.
 \end {equation} 
 From \eqref {ag10} adding and subtracting $\WW((u^{+,K},\Lambda_n), (u^{+,K},\Lambda_n^c))$
 we get
  \[    G_1^{u^{+,K}}(u^{+,K},\Lambda_n) \ge \KK_1(v,\Lambda_n)  +  \WW((u^{+,K},\Lambda_n), (u^{+,K},\Lambda_n^c))  + c \eta/2 (2n)^d . \]
  Taking into account \eqref  {ag11} we obtain
  \[ G_1^{u^{+,K}}(u^{+,K},\Lambda_n) \ge G_1^{u^{+,K}}(v,\Lambda_n)   -2\eps|\Lambda_n|  + c \eta/2 (2n)^d . \]
 Choosing $n$ sufficiently large  we get 
$G_1^{u^{+,K}}(v,\Lambda_n)<G_1^{u^{+,K}}(u^{+,K},\Lambda_n),$ which contradicts the fact that $u^{+,K}$ is a minimizer under compact perturbations.
Note that the proof works for all $C^+ \ge 1+C_0 \theta \|g\|_\infty,$ which proves \eqref{eq:bound1}. 
\end{proof}

\begin {defin} \label {feb23}  {\bf Infinite volume  states}   Let       $K \in \R$,  $K\ge   1 + C_0 \theta \|g\|_\infty$  and $u^{\pm,K} (\cdot, \om)$    the functions constructed  in  Proposition   \ref {ag1}.
 We define  the infinite volume states  $v^{\pm} (\cdot, \om)$ be
the  following   point-wise   limit:
\begin {equation}     \label {ag12}    \lim_{K \to \infty} u^{\pm, K} =v^{\pm}, \qquad \Pr-a.s.   \end {equation}
The limit is well defined since  $\|u^{\pm,K} (\cdot, \om)\|_\infty \le  1+ C_0 \theta \|g\|_\infty $   and the sequence $ \{u^{+,K} (\cdot, \om)\}_K$ is increasing ($\{u^{-,K} (\cdot, \om)\}_K$ decreasing) in $K$.
\end {defin} 
   In the next lemma  we show  that the  $v^{\pm}$ inherit    the  regularity  of $u^{\pm,K}$ and that convergence in \eqref  {ag12}  holds in a stronger norm. 
  
   \begin{lem} \label {maye3}
   Let       $K \in \R$,  $K\ge   1 + C_0 \theta \|g\|_\infty$  and $u^{\pm,K} (\cdot, \om)$    the functions constructed  in  Proposition   \ref {ag1}.  Then  $v^{\pm} (\cdot, \om)$  defined in \eqref {ag12}   are in $C^{0,\alpha}_{loc} (\R^d) $ for any $\alpha<2s$ for $2s\le 1,$ and 
   $C^{1,\alpha}_{loc}(\R^d)$ for any $\alpha<2s-1$ for $2s> 1$.
    Further  for any   $\La \Subset \R^d$ the convergence in \eqref {ag12}  holds  in $C^{0,\beta}(\La)$,  $ \beta <\alpha< 2s$ when $s \in (0, \frac 12 ],$  and in   $C^{1,\beta}(\La)$,  $ \beta <\alpha< 2s-1$ when $s \in (\frac 12,1)$.
 \end{lem}
  \begin {proof}   
 By  Proposition  \ref {ag1}  $\{u^{\pm,K} (\cdot, \om)\}_K $ is  bounded and in  $C^{0,\alpha}_{loc}$ for any $\alpha<2s$ for $2s\le 1,$ and 
   $C^{1,\alpha}_{loc}$ for any $\alpha<2s-1$ for $2s> 1.$    
 This implies that subsequences converge
locally uniformly to an Holder function of order  $\alpha<2s$    when  $2s\le 1$  and when  $2s> 1$ to a function  in $C^{1,\alpha}_{loc}$  for $\alpha<2s-1$.
  As the entire sequence converges point-wise, the
limit of any subsequence must coincide with  $v^\pm$ , which is therefore a locally Holder continuous function of order  $\alpha<2s$    when  $2s\le 1$  or  when  $2s> 1$   a function  in $C^{1,\alpha}_{loc}$  with $\alpha<2s-1$. 
  From this and the compact embedding of Holder spaces,  see    Remark  \ref {ra2},  we deduce that   $u^{\pm,K} $ converge to $v^\pm $  on any compact set $\La$   in $C^{0,\beta} (\La)$, $ \beta <\alpha<2s$ when  $2s\le 1$   and on  $C^{1,\beta} (\La)$, $ \beta <\alpha$ when  $\alpha<2s-1$. 
 \end {proof} 
 
 The following Lemma states that point-wise limits of minimizers under compact perturbations are minimizers under compact perturbations. As we could not find an appropriate result in the literature, we prove it here in the form needed for this paper. 
  \begin{lem} \label {maye2}
 Let $\Psi^k:\R^d \to\R$ be a family of uniformly bounded (in $L^\infty$) minimizers under compact perturbations  of $G_1$, see Definition \eqref {min0}. 
  Assume that $\{\Psi^k\}$ converges   point-wise to a function $\Psi:\R^d\to \R.$ Then $\Psi$ is a minimizer  of $G_1$ under compact perturbations.
 \end{lem}
  \begin {proof}  In the following      $\omega $ is a parameter, so we avoid to write it explicitly. 
 We show the lemma by contradiction.
    Assume  that $\Psi$ is not a minimizer under compact perturbation.
     Then  there exists a compact set  (which we may assume to be a cube)  $ \Lambda$  and a measurable  function $u$ so that
 $ G_1^{\Psi}(u, \om, \La)<  G_1^{\Psi}(\Psi, \om, \La)$.    
 Denote  $\Lambda_1= \Lambda \cup \{x \in \R^d: dist (x, \La)\le 2\}$
  $$E_1:= G_1^{\Psi}(\Psi, \om, \La), \qquad E_2:=G_1^{\Psi}(u, \om, \La), \qquad E_k:=G_1^{\Psi^k}(\Psi^{k},\om, \Lambda_1). $$ 
By  assumption
  there exists a $\delta>0$ such that $E_2+\delta<E_1.$ The aim is to construct a function $\widetilde{\Psi^k}$, for some   $k$ large enough,  such that  if  $E_2+\delta<E_1 $ then 
  $G_1^{\Psi^k}(\widetilde{\Psi^k},\Lambda_1)<E_k $,  which  gives   a contradiction, since $ \Psi^k$ is by assumption a minimizer under compact perturbations.   The proof is similar to the one in  Proposition \ref {jt2}. 
   
 \noindent 
  \emph{ Step 1}:  By \eqref {funct11}
  \begin{eqnarray}\label{split1} E_1\!\!&=&\!\!\KK_1(\Psi,\Lambda)+\WW((\Psi,\Lambda),(\Psi,\Lambda^c))   \\
  \label{split2} E_2\!\!&=&\!\!\KK_1(u,\Lambda)+\WW ((u,\Lambda),(\Psi,\Lambda^c))  \\
  \label{split3} E_k\!\!&=&\!\!\KK_1(\Psi^k,\Lambda_1)+\WW((\Psi^k,\Lambda_1),(\Psi^k,\Lambda_1 ^c))  \end{eqnarray}
  We write $E_k$ as
$$E_k  =    \KK_1(\Psi^k,\Lambda)+ \WW((\Psi^k,\Lambda),(\Psi^k,\Lambda^c)) +  B_k, $$ 
 where  $$B_k= \KK_1(\Psi^k,\Lambda_1\setminus \La)  + 
  \WW((\Psi^k,\Lambda_1\setminus \La),(\Psi^k,\Lambda_1 ^c)). 
    $$ 
    \vskip0.5cm 
  \noindent
  \emph{ Step 2:}  Next we show that for any $\eps>0$ there exists  $k_\eps$ s.t. for $ k \ge k_\e$
   \begin{eqnarray}
  \label{est1b} |\KK_1(\Psi^k,\Lambda)-{\KK_1}(\Psi,\Lambda)|<\eps,\\
  \label{est2b}
A\equiv  | \WW ((\Psi,\Lambda),(\Psi,\Lambda^c))- \WW ((\Psi^k,\Lambda),(\Psi^k,\Lambda^c))|<\eps.
  \end{eqnarray}
  The   \eqref{est1b}  follows immediately from    Proposition \ref{convest} with $D=\Lambda$,
  the regularity property of the minimizers, see  Lemma \ref {maye3}. 
      For  \eqref {est2b},  fix   $R>0$ so that  $\Lambda \subset B_{R/2}(0)$.
    We  upper bound $A$ in \eqref {est2b} as following:
  \begin{eqnarray*}
 &&  A \le |I_1|+|I_2|, \\
 &&I_1=\int_{\Lambda}\int_{B_R(0) \setminus\Lambda}\frac{|\Psi(z)-\Psi(z')|^2-|\Psi^k(z)-\Psi^k (z')|^2}{|z-z'|^{d+2s}}dz dz',\\
 &&I_2=\int_{\Lambda}\int_{\Lambda^c\setminus B_R(0)}\frac{|\Psi(z)-\Psi(z')|^2-|\Psi^k(z)-\Psi^k(z')|^2}{|z-z'|^{d+2s}}dz dz'.
 \end{eqnarray*}
 $I_1$ is estimated (in a very rough way) by Proposition  \ref{convest} with $D=B_R$. 
For $I_2$,   since  $|\Psi|\le C^+,\ |\Psi^k |\le C^+$   
  we  have
  $$
  |I_2|\le\int_{\Lambda}\int_{\R^d\setminus B_R(0)}\frac{8C^+}{|z-z'|^{d+2s}}dz dz'\le 8C^+C(d)|\Lambda|\int_{R/2}^\infty r^{-2s-1}\le C^+
  |\Lambda|C'(d)R^{-2s}.
  $$
Here we used the  integrability of the kernel at infinity. In conclusion, by choosing first $R$ sufficiently large, depending on $\eps,$ and then $k_\eps$ large
depending on $R$ we obtain  \eqref{est1b} and \eqref{est2b}  for all $k\ge k_\eps.$
 \vskip0.5cm 
\noindent
\emph{ Step 3:} By  \eqref{est1b} and  \eqref{est2b} for  $k$ sufficiently large
   \begin {equation} \label {ag6b}    E_k > E_1-2\eps +B_k  > E_2 + \delta -2\eps +B_k. 
 \end {equation}
 \emph{ Step 4}    Define   a function  $\widetilde{ \Psi^k}$ which is equal to $u$ in $\Lambda$ and equal to $\Psi^k $ outside a boundary layer of width 1 of    $\Lambda $. 
   \begin {equation} \label {ag2b}\widetilde { \Psi^k}(x):=\left\{\begin{array}{ll}u(x),&\ {\rm if\ } x\in \Lambda, \\  \Psi^k (x)&  {\rm if\ }    x\in\R^d:\ {\rm dist}(x,\Lambda)>1 \\  \Psi (x)+\Phi(x)(\Psi^k (x)-\Psi (x)) & {  \rm else} 
  \end{array}\right.
  \end {equation} 
   where 
  $\Phi:\R^d\to [0,1]$ is a   smooth cut-off function nondecreasing in ${\rm dist}(x, \Lambda)$  with $\Phi(x)=0$ if ${\rm dist}(x,\Lambda)<1/2$ and 
$\Phi(x)=1$ if ${\rm dist}(x,\Lambda)>1.$  
      Then 
 \begin{eqnarray*}
 &&I_3:=|\WW ((u,\Lambda),( \Psi,\Lambda^c)) -\WW ((u,\Lambda),(\widetilde{\Psi^k},\Lambda^c))|\\&&=\left|\int_{\Lambda } {\rm d} z \int_{\Lambda^c} {\rm d} z' \frac{|u(z)- \Psi(z')|^2-|u(z)-\widetilde{\Psi^k} (z') | ^2}{|z-z'|^{d+2s}}\right|\\&& =\left |\int_{\Lambda } {\rm d} z\int_{\Lambda^c}  {\rm d} z' \frac{2    u(z) [ \widetilde{\Psi^k}(z')-\Psi (z')] +   [ \widetilde{\Psi^k}(z')- \Psi(z')]   [ \widetilde{\Psi^k}(z')+ \Psi(z')] 
}{|z-z'|^{d+2s}}\right |  \\&&=
\left |\int_{\Lambda } {\rm d} z \int_{\Lambda^c}  {\rm d} z'  \1_{|z-z'|>1/2} |\frac{2    u(z) [ \widetilde{\Psi^k}(z')-\Psi (z')] +   [ \widetilde{\Psi^k}(z')- \Psi(z')]   [ \widetilde{\Psi^k}(z')+ \Psi(z')] 
}{|z-z'|^{d+2s}}\right |. 
\end{eqnarray*}
The last equality holds since  $\widetilde{\Psi^k}(x)= \Psi(x)$ for $x \in \Lambda^c$ and  ${\rm dist}(x,\Lambda)<1/2,$ therefore the  integrand vanishes unless $|z-z'|>1/2.$  
Take   $R$ so large  that $ \La \subset B_{\frac R 2} (0)$   and split $\Lambda^c=(\Lambda^c\cap B_R(0))\cup (\Lambda^c\setminus B_R(0))$. We obtain 
 $$
 I_3\le C(d)|\Lambda|R^d\| \Psi- \Psi^k \|_{ L^\infty(B_R)}+|\Lambda|C(d, \theta, C_0, \|g\|_\infty) R^{-2s}.
 $$  For any $\e$ take   $ R_0(\e)$  so that for $R \ge R_0(\e)$     $|\Lambda|C(d, \theta, C_0, \|g\|_\infty) R^{-2s}. \le \frac \e 2$,       then take $K_0  $ depending on   $\epsilon,$ so that  that for $K \ge K_0$,  $|I_3|<\epsilon$ and hence
  \begin {equation} \label {ag3b} E_2 = {\mathcal K}_1(u, \Lambda)+\WW ((u,\Lambda),( \Psi,\Lambda^c)) \ge {\mathcal K}_1(u, \Lambda)+\WW ((u,\Lambda),(\widetilde{\Psi^k},\Lambda^c))-\epsilon.  
 \end {equation} 
By definition of $\widetilde{ \Psi^k}$
$$\WW ((u,\Lambda),(\widetilde{\Psi^k},\Lambda^c))= 
 \big[\WW ((u,\Lambda),(\widetilde{\Psi^k},\Lambda_1 \setminus\Lambda))+\WW_1
 (u,\Lambda),(\Psi^k, \Lambda_1^c)\big], $$  we therefore obtain
  \begin {equation} \label {ag5b}
 E_2\ge {\mathcal K}_1(u, \La)+ \big[\WW ((u,\Lambda),(\widetilde{\Psi^k},\Lambda_1 \setminus\Lambda))+\WW_1
 (u,\Lambda),(\Psi^k, (\Lambda_1)^c)\big] - \epsilon.
\end  {equation}
\emph{ Step 5}
  By  \eqref {ag2b} and \eqref {ag5b}
\begin{equation}  \label {mae1} \begin {split}  &
G_1^{\Psi^k}(\widetilde{\Psi^k},\Lambda_1) = 
{\mathcal K}_1( u, \Lambda )+\WW ((u,\Lambda),(\widetilde{\Psi^k},\Lambda_1\setminus\Lambda))+\WW(
 (u,\Lambda),(\Psi^k, \Lambda_1^c))\\&  +{\mathcal K}_1(\widetilde{\Psi^k}, \Lambda_1\setminus\Lambda)+\WW(
 (\widetilde{\Psi^k},\Lambda_1\setminus\Lambda),(\Psi^k,\Lambda_1^c))\\&\le  E_2+ \eps+
 {\mathcal K}_1(\widetilde{\Psi^k}, \Lambda_1  \setminus\Lambda)+\WW(
 (\widetilde{\Psi^k}, \Lambda_1\setminus\Lambda),(\Psi^k, \Lambda_1^c)). 
\end {split} \end{equation} 
Next we   show that for any $\e>0$ there exists $k_\e$ so that for $k \ge k_\e$
\begin{eqnarray}\label {est3b}
&&\left| {\mathcal K}_1(\widetilde{\Psi^k}, \Lambda_1\setminus\Lambda) - \KK_1(\Psi^k,  \Lambda_1\setminus\Lambda)\right|<\eps\\ 
\label{est4b}
&&\left| \WW(
 (\widetilde{\Psi^k},\Lambda_1\setminus\Lambda),(\Psi^k, \Lambda_1^c)) - \WW ((\Psi^k, \Lambda_1\setminus\Lambda),(\Psi^k, \Lambda_1^c))\right|<\eps.
\end{eqnarray}
Assuming that  \eqref {est3b} and  \eqref {est4b} hold, we obtain 
from  \eqref  {ag6b}  and \eqref {mae1}  that 
$$
E_k>  E_2 + \delta -2\eps +B_k \ge  -4\epsilon+\delta+G_1^{\Psi^k}(\widetilde{\Psi^k},\Lambda_1)
$$for $k$ sufficiently large.   As $\eps$ was arbitrary   and $E_k$ is  minimal  value with $\Psi^k$-boundary conditions,   $\delta=0$
and hence $\Psi$ is a minimizer under compact perturbations.

To prove \eqref{est3b}, we notice that 
    $\widetilde {\Psi^k} (x)=  \Psi(x)+\Phi(x)(\Psi^k(x)- \Psi(x))   =\Psi^k(x)+(1-\Phi(x))( \Psi(x)-\Psi^k(x)) $   and $\Phi (x)=1$ when ${\rm dist}(x, \La) \ge 1$   and  $\| \widetilde{\Psi^k}-\Psi^k \|_{C^{0, \beta}  (\La_1 \setminus \La)} \to 0 $
 for $\beta <\alpha <2s$ when $ s \in (0, \frac 12]$ and $\|\widetilde{\Psi^k}-\Psi^k \|_{C^{1, \beta} (\La_1 \setminus \La)} \to 0$ 
for $\beta <\alpha$ when $ s \in  ( \frac 12,1)$. Therefore by Proposition \ref {convest} for $k$ large enough 
$$ \left|{\mathcal K}_1(\widetilde{\Psi^k}, \Lambda_1\setminus\Lambda) - \KK_1(\Psi^k,  \Lambda_1\setminus\Lambda)\right| \le \e
  $$
 Note that the difference is equal to zero for
${\rm dist}(x,\Lambda)>1.$    
Next we  prove  \eqref{est4b}. We have 
\begin{eqnarray*}
&&\int_{(\Lambda_1\setminus\Lambda) \times\Lambda_1^c}\frac{\big||\widetilde {\Psi^k} (z)-\Psi^k(z')|^2-|{\Psi^k}(z)-\Psi^k (z')|^2\big|}{|z-z'|^{d+2s}}  \\&&=
\int_{(\Lambda_1\setminus\Lambda) \times\Lambda_1^c} \1_{\{ dist (z, \La) \le 1\}} \frac{\big||\widetilde {\Psi^k} (z)-\Psi^k(z')|^2-|{\Psi^k}(z)-\Psi^k (z')|^2\big|}{|z-z'|^{d+2s}} 
 \\&& \le  C 
 \int_{(\Lambda_1\setminus\Lambda) \times\Lambda_1^c}  \1_{\{ dist (z, \La) \le 1\}}  \frac{ |\Psi(z)-\Psi^k(z)|}{|z-z'|^{d+2s}}\\ &&\le C
|\Lambda|^{\frac {d-1} d}  \|\Psi -\Psi^k \|_{L^\infty(\Lambda_1)} \int_1^\infty r^{-2s-1}dr\le
 C|\Lambda|^{\frac {d-1} d}  \|\Psi -\Psi^k \|_{L^\infty(\Lambda_1)} \le \e
\end{eqnarray*}
if $k \ge k_\e$. 
 
 \end{proof} 
 Now we can prove the main theorem:

{\bf  Proof of Theorem \ref {infvol}}  Let $v^\pm$  be the infinite volume  states defined  in \eqref {ag12}. The existence and the first  three  properties of $v^{\pm}$ are established in  Proposition  \ref {ag1}
for $ u^{\pm, K}$ and they are   inherited by the limit.  
 Lemma \ref  {Ma1} establishes the  $L^\infty$ bound for $ u^{\pm, K}$ which  is inherited by the limit as well. 
  The proof  that $v^{\pm}$ are minimizers under compact perturbation   is done in Lemma \ref {maye2}.   Next we  prove  \eqref {M1a}.  We have 
\begin {equation}  \label {M2} \begin {split} 
 & \int_{\Lambda_n}v^\pm(x, \omega) {\rm d} x=  \sum_{z \in \Lambda_n \cap \Z^d} \int_{\{ z+ [-\frac 12, \frac 12]^d\}} v^\pm(x, \omega) {\rm d} x \cr &
 = \sum_{z \in \Lambda_n \cap \Z^d} \int_{   [-\frac 12, \frac 12]^d} v^\pm(T_z x, \omega) {\rm d} x = \sum_{z \in \Lambda_n \cap \Z^d} \int_{   [-\frac 12, \frac 12]^d} v^\pm(x, T_{-z}\omega) {\rm d} x. 
\end {split} \end {equation} 
Since $|v^\pm(x,  \omega)| \le (1+C_0\theta \|g\|_\infty) $,  by the  Birkhoff's ergodic theorem, see for example  \cite{GK},   we have $\Pr-$ a.s 
\begin {equation}  \label {M3aa} \begin {split} 
 & \lim\frac 1 {n^d}\int_{\Lambda_n}v^\pm(x, \omega) {\rm d} x =  \lim \frac 1 {n^d} \sum_{z \in \Lambda_n \cap \Z^d}  
\int_{   [-\frac 12, \frac 12]^d}     v^\pm(x, T_{-z}\omega)  {\rm d} x \cr &
= \E \left [  \int_{   [-\frac 12, \frac 12]^d}   v^\pm(x, \cdot )  {\rm d} x \right ] = m^{\pm}.
 \end {split} \end {equation}

It remains to show \eqref{diseq1}. Let $  \bar w_n$ be as in the statement of the theorem and
 fix  $x\in\Lambda_n.$ Denote $K=  \max \{ \|\bar v_0\|_\infty, (1+C_0\theta \|g\|_\infty) \}$. 
Let $u_n^{\pm,K}$ the $K-$ maximal and the $K-$ minimal minimizer of $G_1$ in $\Lambda_n$,
see Definition \ref {exmin}. 
By Lemma \ref{FGK}
we get that $u_n^{-,K}(x,\omega)\le \bar w_n(x,\omega)
\le u_n^{+,K}(x,\omega)$  for $x \in \R^d$. Then, by \eqref {ag8},  uniformly for any   compact set of $\R^d$ containing $x$ we have 
$$ v^{-}(x,\omega) \le   u^{-,K}(x,\omega)  \le    \liminf_n \bar w_n(x,\omega) \le  \limsup_n \bar w_n(x,\omega) \le u^{+,K}(x,\omega)  \le v^{+}(x,\omega).$$
The first and last inequality hold since   
 $ \{u^{+,K}\}_K$ is increasing ($\{u^{-,K}\}_K$ decreasing) in $K$.
The  \eqref{diseq1} follows. 
\qed

In the next Lemma we  bound  uniformly in   $ \om$ the   difference   between the energy of the two extrema macroscopic minimizers  $v^{\pm} $.

 \begin{lem}\label {A1}  Let     $ \Lambda \Subset \R^d$, cube-like,   $v^{\pm} $ be  the     infinite volume   states  constructed in Theorem \ref {infvol}. 
 There exists a positive  constant $C$   depending on $\theta$, $d$, $s$,  $C_0$ and $\|g\|_\infty$,
  so that $\Pr-$ a.s.
    \begin{equation} \label{m3}\left |  G_1 (v^+ , \om,\La) - G_1 (v^-, \om, \La)\right | \le  \left \{ \begin {split} & C | \Lambda|^{\frac  {d-2s} d},   \quad s \in (0, \frac 12),
      \cr &   C | \Lambda|^{\frac  {d-1} d} ,  \quad \quad s \in (\frac 12, 1), \cr &  
   C   |\La|^{\frac {d-1} d}  \log | \Lambda| , \qquad s = \frac 12 . \end {split} \right.
\end{equation}    
 \end{lem}
 \begin {proof} 
 Let the cut-off function 
  $\Psi: \R^d \to \R$ be a smooth  nondecreasing  function in $\dist (x, \La^c)$  with $\Psi(x)=1$ if $\dist (x, \La^c) \ge 1 $ and 
$\Psi(x)=0$ if $\dist (x, \La^c)=0.$  
   Set 
 \begin{equation} \label{gt1} \tilde u:=\Psi   v^+  + 
\left(1-\Psi\ \right) v^-. 
\end{equation}   The function  $\tilde u$ is  equal to $v^-$  when $x \in  \Lambda^c$ and     equal to $ v^+ $  when $x \in \La$,  $\dist (x, \La^c) >1$  and interpolates in a smooth way  between these values.
Since $v^-$ is the  minimal $-$ minimizer in $\La$ we have
    \begin{equation}  \label{s1}   
     G_1(v^{-} , \om, \La)  \le   G_1^{v^-}(\tilde u , \om, \La). \end{equation}  
    We will show that 
    \begin{equation}  \label{s1g}  G_1^{v^-}(\tilde u , \om, \La) \le   G_1 (  v^{+} , \om, \La)  + M(s)  \end{equation}  
    where we denote shortly by $M(s)$ the right hand side of \eqref {m3}. 
    Therefore  from \eqref {s1}
    \begin{equation} \label{g1c}  G_1 (v^{-} , \om, \La) -    G_1 (  v^{+} , \om, \La) \le  M(s).
     \end{equation} 
    In a similar way we can  show that
   \begin{equation} \label{g2} G_1 (  v^{+} , \om, \La)- G_1(v^{-} , \om, \La)   \le  M(s).
     \end{equation} 
    Then,  from  \eqref {g1c} and \eqref {g2} we  get \eqref {m3}. 
    Next we show \eqref {s1g}. 
     By definition
  \begin{equation}  \label{s1b}   
        G_1^{v^-}(\tilde u , \om, \La)  =  \KK_1(\tilde u , \om, \La)  +  \WW((\tilde u, \La) (v^-, \La^c)). \end{equation} 
        Denote by $$ \partial_{\La}= \{ x \in \La: \dist(x, \La^c)\le 1 \}. $$
By definition of $\tilde u$, see \eqref {gt1},  we have 
       \begin{equation} \label{s2}     \KK_1 (\tilde u , \om, \La)      =
        \KK_1 (v^+ , \om, \La \setminus   \partial_{\La})  +  \KK_1 (  \tilde u , \om,     \partial_{\La})  + \WW (  ( v^+ , \La \setminus     \partial_{\La}),  (\tilde  u,     \partial_{\La}) ).  \end{equation}    
     By adding and subtracting   $\KK_1 (  v^+ , \om,     \partial_{\La})$ and  the interaction term   $ \WW (  ( v^+ , \La \setminus   \partial_{\La} ),  ( v^+,     \partial_{\La}) )$ we  get 
         \begin{equation}  \label{s3} \begin {split}  \KK_1 (\tilde u , \om, \La) &=
 \KK_1 (v^+, \om, \La) +  \left [ \KK_1 (  \tilde u , \om,    \partial_{\La}) - \KK_1 (  v^+ , \om,    \partial_{\La})\right]  \cr & +  \left [ \WW (  ( v^+ , \La \setminus     \partial_{\La} ), ( \tilde u,      \partial_{\La}) )-  \WW (  ( v^+ , \La \setminus    \partial_{\La} ), (v^+,     \partial_{\La}) )\right]. 
\end {split}  \end{equation}    
For the second term of \eqref {s1b}    we add and subtract  $ \WW( (v^+, \La), (v^+, \La^c))$  obtaining 
\begin{equation}  \label{s4}     \WW((\tilde u, \La) (v^-, \La^c))   =  \WW( (v^+, \La), (v^+, \La^c)) + \left [  \WW((\tilde u,\La),(v^-, \La^c))-   \WW( (v^+, \La), (v^+, \La^c))  \right ]. \end{equation}   
Taking into  account
  \eqref {s1b}, \eqref {s3}  \eqref {s4} we get that
   \begin{equation}  \label{s6a}   G_1^{v^-}(\tilde u , \om, \La)  =  G_1^{v^+}(v^+ , \om, \La) + \RR_1 +    \RR_2 +\RR_3\end{equation}    
where    
  \begin{equation}  \label{s6} \begin {split}   \RR_1&=
   \left [ \KK_1 (  \tilde u , \om,   \partial_{\La}) - \KK_1 (  v^+ , \om,      \partial_{\La} )\right],  \cr  
   \RR_2  &=   \left [ \WW (  ( v^+ , \La \setminus     \partial_{\La} ), ( \tilde u,       \partial_{\La}) )-  \WW (  ( v^+ , \La \setminus   \partial_{\La} ), (v^+,      \partial_{\La}) )\right], \cr  
\RR_3&=\left [  \WW((\tilde u,\La),(v^-, \La^c))-   \WW( (v^+, \La), (v^+, \La^c))  \right ].
 \end {split}  \end{equation}     
    Since  $ \RR_2$ and $ \RR_3$ are difference of positive terms and     $\tilde u$, $v^-$ and  $v^+$ are smooth enough  we can apply  \eqref {t3c} of Lemma \ref {boun1} to each single term   obtaining 
  $$ \left | \RR_2 \right|  \le  M(s), \qquad  \left |  \RR_3 \right | \le M(s). $$ 
 
Next we  estimate  $ \RR_1$.
   We have
       \begin{equation}    \label{s8}\begin {split}   & \left | \RR_1  \right |   \le 
  \int_{  \partial_{\La} }  dx \int_{  \partial_{\La} } dy \frac {  \left | 
 ( \tilde u(x) - \tilde u(y) )^2  -   (  v^+(x) - v^+(y) )^2  \right |}  { |x-y|^{d+2s}} \cr &   + \int_{  \partial_{\La}} 
 \left | W(\tilde u(x))  -   
 W(v^+(x))  \right |  \rm {d }x 
+   \theta  \int_{   \partial_{\La} }   \left | g_1 (x,\omega) \left [  \tilde u(x)- v^+(x)  \right ]  \right | \rm {d }x \cr  & \le 
 \int_{   \partial_{\La} }  dx \int_{  \partial_{\La} } dy \frac {  \left | 
 ( \tilde u(x) - \tilde u(y) )^2  -   (  v^+(x) - v^+(y) )^2  \right | }  { |x-y|^{d+2s}} \cr &   + 
C(C_0,\theta, \|g\|_\infty) |\La|^{\frac {d-1} d}   
\end {split}  \end{equation}   
where $C(C_0,\theta, \|g\|_\infty)$ is a constant which depends only on $\theta$, the bound on the random field, see \eqref {eq:ass}  and the interaction $W$.
We need some care to estimate the integral term in \eqref {s8} since the integral  might be singular.
We exploit the regularity of the minimizers.  Recall   that  for   $s \in (0, \frac 12] $,  $v^+ \in C^{0,\alpha}_{loc} (\R^d)$  for $\alpha  <2s$  and 
for  $s \in (\frac 12,1)  $,  $v^+ \in C^{1,\alpha} _{loc} (\R^d) $  for $\alpha  <2s -1$.  The same regularity holds by construction for $ \tilde u$. 
Therefore     
   \begin{equation}    \label{s8g}\begin {split} &
 \int_{   \partial_{\La} }  \int_{   \partial_{\La} } \frac {  \left [ 
 ( \tilde u(x) - \tilde u(y) )^2  -   (  v^+(x) - v^+(y) )^2  \right ] }  { |x-y|^{d+2s}}  \cr & \le  \left \{ \begin {split} &
  2 C \int_{    \partial_{\La} }  \int_{   \partial_{\La} } \frac1  { |x-y|^{d+2s -2\alpha}} \quad   s \in (0, \frac 12]  \cr &
   2 C \int_{   \partial_{\La} }  \int_{    \partial_{\La}} \frac1  { |x-y|^{d+2s -2}}  \quad     s \in ( \frac 12, 1). \end {split} \right. 
  \end {split}  \end{equation}  
  We have that when $ s \in (0, \frac 12]$,   $2s -2\alpha  <0 $    
  and  when $ s \in (\frac 12,1)$,   $2s -2 <0 $. 
  Therefore both terms on the right hand side of \eqref {s8g} are integrable and bounded by 
  $ C | \Lambda|^{\frac  {d-1} d} $.

    \end {proof} 

\vskip0.5cm \noindent 
The  quantity  defined  next plays a fundamental  role. 
  \begin{defin}  \label{def1}{\qquad } 
 \begin{enumerate}
 \item For a cube  $\Lambda\subseteq \R^n$ we define ${\mathcal B}_\Lambda$ as the $\sigma$-algebra generated by the random field in $\Lambda.$
 \item  
Let  $v^\pm (\om)$ be  the infinite volume states constructed before.    We define
  \begin{equation} \label{m4}   F_n (\om):=    \E \left [  \left \{G_1(v^+(\cdot), \cdot, \La_n) - G_1(v^-(\cdot), \cdot, \La_n) \right \} | \BB_{\La_n} \right ]. 
 \end{equation}   
   \end{enumerate}
 \end {defin} 
 \vskip0.5cm \noindent 
 \begin {rem}  By definition $F_n (\cdot)$ is $\BB_{\La_n}$ measurable and by the symmetry  assumption on the random field  $\{g(z,\cdot), z \in \Z^d\}$ 
 \begin{equation} \label{m9} \E\left [  F_n (\cdot)\right ] =0. \end{equation} 
 Namely   $ v^+ (x, \om) = -v^- (x, -\om) $ for   $x \in \R^d$.  This implies that 
  \begin{equation} \label{m8a}    G_1(v^+(\om), \om, \La_n)= G_1(v^-(-\om), -\om, \La_n) 
  \end{equation}
 and by  the symmetry of the random field  we get  \eqref {m9}. 
\end {rem}
Next we want to  quantify  how   $v^\pm (\om)$ changes  when  the random field is modified  only in one site, for example at the site $i$.   We introduce the following notation:  
$$  \om^{(i)}: \om^{(i)}  (z)= \om (z) \quad z \neq i, \qquad   \om= (\om(i),  \om^{(i)})   \quad i,z \in \Z^d. $$
The $v^+(\cdot, (\om(0), \om^{(0)}))$ is then  the state $v^+$  when  the 
random field  at the origin is   $\om(0)$,  and   $v^+(\cdot, (\om(0)-h, \om^{(0)}))$   the state $v^+$ when the 
random field at the origin is   $\om(0)-h$, and the same definition is used for the infinite volume  state $v^-(\cdot, (\cdot, \om^{(0)}))$ and  for  the finite volume minimizers      $v^\pm_n(\cdot,  (\cdot, \om^{(0)}))$.

Now we are able to state the
following lemma:  

\begin{lem}   \label {A2}   For    $ \La \Subset  \R^d$, $ 0 \in \La$,      $h>0$ we have 
 \begin{equation} \label{LL.4}  \begin {split}
  \theta h\int_{Q (0)}v^+(\om(0), \om^{(0)}) {\rm d} x  & \ge
 G_1(v^+(\om (0)-h,\om^{(0)}),(\om(0)-h, \om^{(0)}), \La)- G_1(v^+(\om(0),\om^{(0)}),(\om (0), \om^{(0)}), \La) \cr & \ge
 \theta h \int_{Q(0)}v^+(\om (0)-h, \om^{(0)} )  {\rm d} x  \end {split}
 \end{equation}
where   $Q(0)=[-1/2,1/2]^d$.
The same  inequalities hold for   $v^-$. 
\end{lem}
 \begin {proof} Let $ \La_n$ be  a cube centered at the origin  so that $ \La \subset \La_n$, $ K \ge (1+C_0\theta \|g\|_\infty)$.   Let  $v^+_n  = v^{+,K} $ be the  $K-$ maximal minimizer of $G_1$ in $\Lambda_n$  see Definition \ref {exmin}.  
 Remark  that $v^+_n $   is measurable with respect to  the random field $g (z,\om)$,   $ z \in \La_n \cap \Z^d$. 
We have
\begin{equation} \begin{split}\label{v3b}  &
   G_1(v^+_n(\om (0),\om^{(0)}),(\om(0), \om^{(0)}), \La)
-G_1(v^+_n(\om (0)-h,\om^{(0)}),(\om(0)-h, \om^{(0)}), \La) \cr &=
  G_1(v^+_n(\om (0),\om^{(0)}),(\om(0), \om^{(0)}), \La)-   G_1(v^+_n(\om (0),\om^{(0)}),(\om (0)-h, \om^{(0)}), \La)\cr &+
  G_1(v^+_n(\om(0),\om^{(0)}),(\om(0)-h, \om^{(0)}), \La)- G_1(v^+_n(\om(0)-h,\om^{(0)}),(\om(0)-h, \om^{(0)}), \La).  
 \end {split}  \end{equation}
By explicit computation,  see \eqref{functional2}, we have  that
$$ G_1(v^+_n(\om (0),\om^{(0)}),(\om (0), \om^{(0)}), \La)-   G_1(v^+_n(\om (0),\om^{(0)}),(\om(0)-h, \om^{(0)}), \La)= -  h\theta  \int_{Q(0)}v^+_n (\om (0), \om^{(0)} ) dx. $$
The  last line in \eqref{v3b}    is nonnegative, because $v^+_n(\om(0)-h,\om^{(0)})$
is a minimizer of $G_1$ in $\La_n$  when  the random field  is $(\om(0)-h, \om^{(0)} )$. Therefore 
 $$G_1(v^+_n(\om(0)-h,\om^{(0)}),(\om(0)-h, \om^{(0)}), \La) - G_1(v^+_n(\om (0),\om^{(0)}),(\om (0), \om^{(0)}), \La)\le
    h \theta  \int_{Q(0)}v^+_n (\om (0), \om^{(0)} ) dx .$$
By splitting   
   \begin{equation*} \begin{split}  &
   G_1(v^+_n(\om (0),\om^{(0)}),(\om (0), \om^{(0)}), \La)
-G_1(v^+_n(\om (0)-h,\om^{(0)}),(\om (0)-h, \om^{(0)}), \La) \cr &=
  G_1(v^+_n(\om (0),\om^{(0)}),(\om (0), \om^{(0)}), \La)-   G_1(v^+_n(\om (0)-h,\om^{(0)}),(\om (0), \om^{(0)}), \La)\cr &+
  G_1(v^+_n(\om (0)-h,\om^{(0)}),(\om(0), \om^{(0)}), \La)- G_1(v^+_n(\om (0)-h,\om^{(0)}),(\om (0)-h, \om^{(0)}), \La)
 \end {split}  \end{equation*} we obtain in a similar way
 $$G_1(v^+_n(\om (0)-h,\om^{(0)}),(\om (0)-h, \om^{(0)}), \La) - G_1(v^+_n(\om(0),\om^{(0)}),(\om (0), \om^{(0)}), \La)\ge
   h \theta  \int_{Q (0)}v^+_n (\om(0)-h, \om^{(0)} ) dx .$$
   To pass to the limit 
note  that the cube $Q (0)$ remains fixed. 
Denote by $M$
the smallest integer such that $\Lambda\subseteq B_M(0)$, where $ B_M(0)$ is a ball centered at the origin of radius $M$.   
   
By the smoothness of the minimizers, see Proposition  \ref {Lip}  $  v^+_n \in C^{0, \alpha} (B_M(0)) $     with $ \alpha< 2s$ when $2s<1$  and in $ C^{1, \alpha}(B_M(0)) $, $ \alpha < 1-2s$  when $ s \in [\frac 12, 1)$.  Further the sequence  $\{ v^+_n\}_n $ uniformly converges to $v^{+,K}$ in $B_M(0)$ and  $|v^{+,K}| \le 1+C_0 \theta  \|g\|_\infty$
    uniformly  in $n$ and $K$.   By
Lebesgue's  Theorem on dominated convergence, 
we may pass to the limit under the integral as $n \to \infty$.  By Definition \ref {feb23}  $\{v^{+,K}\}_K $ point-wise converges to $v^+$ when $K  \to \infty$ then  applying again the Lebesgue's  Theorem on dominated convergence we pass to the limit as $K\to \infty$  and the claim is shown.
The  corresponding statement for
$v^-$ is  proved in the same way.
 \end{proof}  
 \begin{rem} \label {R1}  From  Lemma \ref {A2} we have that
 $$ \om(0)\mapsto \int_{Q(0)}v^+(\om(0), \om^{(0)}) {\rm d} x 
 $$ is nondecreasing.

 \end{rem}

\begin{cor}\label{A2b} Let  $\om(i)$ be the random field in the site $i$ which has  probability distribution absolutely continuous w.r.t the Lebesgue measure. 
We have that  $G_1(v^+(\om),\om, \La)$ is ${\mathbb P}$-a.e. differentiable w.r.t to $\om(i)$ and 
$$
\frac {\partial G_1(v^\pm(\om),\om, \La) } {\partial {\om (i)} } = -\theta \int_{Q(i)}v^\pm(x,\om) {\rm d} x.
$$
\end{cor} 
\begin {proof}It is sufficient to consider the case $i=0.$
By applying Lemma \ref{A2} for $\om(0)$ and $\tilde \om(0)=\om(0)+h$ we see that left and right derivatives
exist and are equal if
$
s\mapsto \int_{Q(0)}v^+(s, \om^{(0)}) {\rm d} x 
$ is continuous at $s=\om(0).$ By Remark \ref{R1} this happens for Lebesgue almost all $s,$ hence by
the assumptions on the random field ${\mathbb P} $-a.e.
\end {proof}
\begin {rem}  \label {AC1} When the distribution of $g$ is not absolutely continuous with respect to Lebesgue measure   Corollary \ref {A2b} does not hold.  We still can show Lemma \ref {A2}  but  we can only estimate  from above and below the difference in  the energy  which appears when the  random field is modified in one site.
\end {rem}
  \begin{thm}\label{A3}  Let $ F_n (\cdot )$ be defined in \eqref {m4}, 
  we have that   \footnote {  $ \lim_{ n \to \infty} X_n  \stackrel {D} {=}  Z$ denotes convergence in distribution of the random variables $X_n$ to a random variable $Z.$} 
\begin {equation} \label {MM1} \lim_{ n \to \infty} \frac 1 {\sqrt {| \Lambda_n|}} \left [  F_n (\cdot )  \right ]  \stackrel {D} {=}  Z,     \end {equation}
where $Z$ stands for a  Gaussian   random variable with mean $0$ and variance 
$b^2$, defined in \eqref {D2}  with 
 \begin {equation} \label {MM2}  4 \theta^2 (1+ C_0 \theta \|g\|_\infty)^2 \ge   b^2 \ge  D^2    \end {equation}
where  \begin{equation} \label{may1} D^2 =  \E \left [ \left ( \E   \left [ F_n | \BB(0)\right ] \right )^2 \right ], \end{equation} 
  $ \BB(0)$ is  the  sigma -algebra generated by $g(0, \om)$ and $C_0$ is given in \eqref {V.1}.
 \end{thm}
The  proof of this  theorem is done  invoking the  general result  presented in the appendix and proceeding in the same way as in \cite  {DO2}.  To facilitate the reader we recall below  the main steps 
of the proof. 

\begin {proof}  
   We decompose $  F_n $ as a martingale difference sequence. We 
  order the points in $\La_n \cap \Z^d$  according to the lexicographic ordering.   In the following 
    $ i \le  j$     refers  to the lexicographic ordering. 
  Any other ordering will be fine but it is convenient to fix one.
 We  introduce the family of increasing $\sigma-$ algebra 
 $  \BB_{n,i}$,  $ i \in  \La_n \cap \Z^d $ where 
 $  \BB_{n,i}$ is the $\sigma-$ algebra  generated by the random  variables $ \{g (z), z\in \La_n \cap \Z^d, z \le i  \}  $.  We denote by 
 $$  \BB_{n,0}= (\emptyset, \Omega),  \quad     \BB_{n,i} \subset   \BB_{n,j}   \qquad  i \le j, \quad i \in     \La_n \cap \Z^d,  \quad j \in  \La_n \cap \Z^d.  $$ 
           We split 
   \begin{equation} \label{m91}     F_n  = \sum_{i \in \Z^d \cap \La_n } \left ( \E[ F_n| \BB_{n,i}] - \E[ F_n | \BB_{n,i-1}]\right ):=  \sum_{i \in \Z^d \cap \La_n }  Y_{n,i}.    \end {equation} 
        By construction     $ \E \left [  Y_{n,i}\right ]= 0$ for $i \in \Z^d \cap \La_n $, 
        $ \E \left [  Y_{n,i} | \BB_{n,k} \right ]= 0$, for all  $0 \le  k \le i-1$.
  Denote  
 \begin{equation}  \label{v2}  
V_n: = \frac 1 {   |\La_n \cap  \Z^d|}   \sum_{i \in \La_n \cap  \Z^d} \E \left [  Y^2_{n,i}  | \BB_{n,i-1} \right ] .    \end {equation}  
By  Lemma \ref {d2}  stated  below we have  that  $ V_n \to b^2$  in probability  and $b^2$  satisfies \eqref{MM2}.   
By  Lemma \ref {LL1}   stated   below we  have  that   for any $a>0$ 
 \begin{equation}  \label{v3}    U_{n} (a): = \frac 1 {   |\La_n \cap  \Z^d|}   \sum_{i \in   \La_n \cap  \Z^d }  \E [  Y^2_{n,i} 1_{\{  |Y_{n,i}| \ge  a  \sqrt { |\La_n \cap  \Z^d|}\}}  |  \BB_{n,i-1} ]   \end {equation} 
converges  to $0$ in probability.  
We can then invoke Theorem 5.1, stated in the appendix.   
   The correspondence to the   notation used in the appendix is the following. Identify $|\La_n \cap  \Z^d|$ with $n$,  
 $    \frac {F_n} {\sqrt { |\La_n \cap  \Z^d|} } \leftrightarrow  S_n $,  $\frac {Y_{n,i} }  {\sqrt { |\La_n \cap  \Z^d|} } \leftrightarrow   X_{n,i}$ and $\BB_{n,i}  \leftrightarrow  \FF_{n,i} $.  
 Then \eqref {MM1} is obtained. 
    \end   {proof}
 Before stating Lemma \ref {d2} it is convenient to introduce  a
 new sigma-algebra $\BB_i^{\le}$ generated by the random fields $ \{ g(z, \om),  z \in \Z^d,  z   \le i \} $ where $\le $ refers to the lexicographic ordering. 
 Define for $ i \in \Lambda_n$
\begin{equation} \label{g1a}  W_i[\omega]=  \E \left [G_1(v^+ (\om), \om,\La_n) - G_1(v^- (\om), \om, \La_n) | \BB_i^{\le}\right ] - \E \left [  G_1(v^+(\om), \om,\La_n) - G_1(v^-(\om), \om, \La_n) | \BB_{i-1}^{\le}\right ].  \end{equation}    
Note that      $W_i$ is a   random variable  depending on   random fields on sites   smaller  or equal than $i$ under   the lexicographic order. In particular it does not depend on the choice of the cube $\La_n$ provided $ i \in  \Lambda_n$. The proof of this last statement  uses that the random field has 
a   distribution continuous with respect to Lebesgue measure.  In particular the proof  relies on   Corollary \ref {A2b} and it is done in \cite [Lemma 4.9] {DO2}. 
 
  \begin {lem}  \label {d2} Let  $V_n$  be the quantity  defined in \eqref {v2}. For all $\delta>0$
 \begin{equation}  \label{D1}  
\lim_{n \to \infty}  \Pr \left [   |V_n  -b^2| \ge \delta \right ] =0,
\end {equation} 
where      $W_0$ is defined in \eqref{g1a}  
\begin{equation}  \label{D2}   b^2=  \E \left [ W_0^2 \right ].
    \end {equation} 
Further 
\begin {equation} \label {MM2a} 4 \theta^2 (1+ C_0 \theta \|g\|_\infty)^2 \ge b^2 \ge \E \left [ \left ( \E   \left [ F_n | \BB(0)\right ] \right )^2 \right ],   \end {equation}
where $C_0$ is given in \eqref {V.1}.
\end {lem} 
 
  \begin{lem} \label{LL1}  Let  $U_n(a)$ defined in \eqref {v3}. For any $a>0$ for any $ \delta >0$ 
   $$ \lim_{n \to \infty} \Pr \left [ U_n(a) \ge \delta \right] =0. $$
\end{lem}
For the  proof of   Lemma \ref  {d2} and  Lemma \ref   {LL1} see  \cite  {DO2}.
  
\begin{lem}   \label {A2c}   For    $ \La \subset  \R^d$, $ 0 \in \La$,     
we have
$$
\frac {\partial} {\partial \omega (0)}  \E \left [ F_n| \BB(0)\right ] =   
  - \theta  \E \left [ \int_{ Q (0)} v^+ (x, \om) dx  | \BB(0) \right ] + \theta  \E \left [ \int_{ Q (0)} v^- (x, \om) dx  | \BB(0) \right ]
  $$
where   $Q (0):=[-1/2,1/2]^d$.
Further 
$$ \E \left [  \frac {\partial} {\partial \omega(0)}  \E \left [ F_n| \BB(0)\right ]  \right ] =   - 2 \theta m^+, $$ 
where $m^+$ is defined in \eqref {M1a}. 
\end {lem}
\begin {proof}  The  proof follows from  Corollary \ref{A2b} after taking conditional expectations.
 Further, by  Theorem \ref {infvol},  we have
 \begin{equation} \label{m90} \begin {split} &  \E \left [  \frac {\partial} {\partial \omega  (0)}  \E \left [ F_n| \BB(0)\right ]  \right ] =  -  \theta \E \left [  \E \left [ \int_{ Q(0)} v^+ (x, \om) dx  | \BB(0) \right ]  \right ]  \cr & +  \theta  \E \left [  \E \left [ \int_{ Q(0)} v^- (x, \om) dx  | \BB(0) \right ]\right ]  =  \theta [- m^+ + m^-]= -2 \theta m^+.
 \end {split}
 \end{equation} 
\end {proof} 
 
  \begin {lem}  \label {d3} If 
\begin{equation} \label{d3a}  \E \left [ \left (  \E \left [ F_n| \BB(0)\right ]\right) ^2\right ] =0 \end {equation} 
 then    $m^+=m^-=0$, see for the definition    \eqref {M1a}. 
\end {lem}
   \begin {proof} 
Denote  $f (\om (0)):=  \E \left [- F_n| \BB(0)\right ] $.   
Set $s= \om (0)$, \eqref {d3a}    can be written as 
 $ \int f^2(s) \Pr(ds) =0.$   
  This implies that 
 $f(s)=0$ for $\Pr$ almost all point of continuity of the distribution $g(0)$.
 By Lemma  \ref {A2c} and by bound  \eqref {eq:bound} in Theorem    \ref {infvol}  we have that $(1+C_0 \|g\|_\infty \theta)\theta  \ge f'(s) \ge 0$ almost
 everywhere.  
If $f(s)=0 $  for $\Pr$ almost all point of continuity of the distribution $g$, then $ f'(s)=0$  for  $\Pr$ almost all point of continuity of the distribution of  $\omega(0)$.   But if  $ f'(s)=0$  then  from Lemma \ref {A2c} we get $m^+=m^-=0$.   
\end {proof}

\vskip0.5cm 
\noindent  { \bf Proof of Theorem \ref {min1} }     
Applying  Theorem \ref {A3} we get the following lower bound on the Laplace transform of 
$F_n(\om)$ defined in Definition \ref {def1}:
 \begin{equation} \label{mars1}
 \liminf_{n \to \infty} \E  \left [ e^{t  \frac {F_n} {\sqrt { \La_n}} } \right ] \ge e^{\frac {t^2 D^2}  2}  \end{equation}   
 where  $D^2$ is defined in \eqref {may1}.
 It is immediate to realize that  \eqref {mars1} and the results stated in   Lemma \ref {A1}  contradict each other in $ d=2$  for all $s \in (\frac 1 2, 1)$ and in $d=1$ for $s \in [\frac 14,1)$ unless $D^2=0$.  On the other hand  when  $D^2=0$,  Lemma \ref {d3} implies 
\begin{equation}\label{integralsequal}
m^+= -m^-=    
\E \left [  \int_{   [-\frac 12, \frac 12]^d}   v^\pm(x, \cdot )  {\rm d} x \right ] =0. 
\end{equation} 
Now \eqref{diseq1} implies that $\Pr$-a.s. 
$v^+(x,\omega)\ge v^-(x,\omega)$ for all $x\in R^d.$ This and 
\eqref{integralsequal} imply that $v^+(x,\omega)=v^-(x,\omega)$ a.s.
  By \eqref {diseq1}  $ \Pr-$ a.s.  and uniformly for any compact of $\R^d$ containing $x$ 
we have that 
   $$ v^- (x, \omega) \le  
\liminf_{n\to\infty}  u^*_n (x , \om) \le\limsup_{n\to\infty}   u^*_n (x , \om) 
\le v^+ (x, \omega). $$
 Since  $ v^- = v^+$,   $ \Pr-$ a.s, we obtain 
that $$\liminf_{n\to\infty}  u^*_n (x , \om)= \limsup_{n\to\infty}   u^*_n (x , \om) = u^*(x , \om) = v^{\pm} (x,\om)$$
uniformly on compact of $x$.  
 The properties of the minimizer stated in  Theorem \ref{min1} therefore follow from
the corresponding properties of $v^\pm,$ see Theorem \ref{infvol}.  
 Further 
   we  have  
   $$\E [v^+(x,\cdot)]= _{\rm symm}- \E [v^{-}(x,\cdot)]=_{\rm unique}
   -\E [v^+(x,\cdot)], \qquad x \in \R^d.$$  
   This implies for any $x \in  \R^d$,    $ \E [v^\pm(x,\cdot)] = \E[u^*(x ,\cdot) ]=0$.


    \qed

 \section {Technical Lemmas} 
 In this section we collect   some lemmas we need to prove the main results. 
 
  \begin{prop} \label{tec1}   For any $\e>0$,     for all $v \in H^\alpha_{loc} \cap L^\infty (\R^d)$, $s< \alpha$,   for all cube $ \Delta$ large enough 
 \begin {equation}   \label {Ot3}     \WW (v,  \Delta) \le \e   |\Delta|.   \end {equation}
   \end{prop}
   \begin {proof} Let $L $ be the edge of $\Delta$. 
 We have for any $R>0$,  $R \le \frac 14 diam (\Delta)$ 
\begin {equation}\begin {split}  \label {Ot1}  & 
     \int_{\Delta} \rm {d }x \int_{  \Delta^c} \rm {d }y\frac { | v (x) - v (y)|^2} {|x-y|^{d+2s}}    \cr & =   \int_{ \{ x \in \Delta: d_{\partial \Delta}(x)  \le R \} } \rm {d }x \int_{  \Delta^c} \rm {d }y\frac { | v (x) - v (y)|^2} {|x-y|^{d+2s}}  +    \int_{\{x \in \Delta:  d_{\partial \Delta}(x)  > R \}} \rm {d }x \int_{  \Delta^c} \rm {d }y\frac { | v (x) - v (y)|^2} {|x-y|^{d+2s}} . 
   \end{split}   \end {equation}
   Assume $ s < \frac 12 $, then  from \eqref {Ot1} and boundedness of $v$
   \begin {equation}\begin {split}  \label {Ot2}  & 
     \int_{\Delta} \rm {d }x \int_{  \Delta^c} \rm {d }y\frac { | v (x) - v (y)|^2} {|x-y|^{d+2s}}    \cr &   \le
       \int_{ \{ x \in \Delta: d_{\partial \Delta}(x)  \le R \} } \int_{  \Delta^c} \rm {d }y\frac { C} {|x-y|^{d+2s}} 
       \rm {d }x    +    \int_{\{x \in \Delta:  d_{\partial \Delta}(x)  > R \}} \rm {d }x \int_{  \Delta^c} \rm {d }y\frac {C} {|x-y|^{d+2s}} \cr & \le
         \int_{ \{ x \in \Delta: d_{\partial \Delta}(x)  \le R \} } \rm {d }x   d_{\partial \La}(x)^{-2s}  + |\Delta|   \int_{ |y| \ge R} \rm {d }y\frac {C} {|y|^{d+2s}} 
        \cr & \le C R^{1-2s} L^{d-1}+ \frac {\e} 2   |\Delta| 
   \end{split}   \end {equation}
if $ R$ large enough, since  $  \int_{ |y| \ge 1} \rm {d }y\frac {C} {|y|^{d+2s}} < C$. Given such $R$ we then choose $L$ so large that $CL^{d-1} R^{1-2s} \le \frac {\e} 2 L^d$. Hence \eqref {Ot3} when  $ s < \frac 12 $.
When $ s \in [\frac 12, 1)$ we can still split as in \eqref  {Ot1} and estimate the  second integral of  \eqref  {Ot1} as done in \eqref {Ot2}. Care needs to be taken to estimate the first integral of  \eqref  {Ot1}. In this case  fix  $ \rho>0$
\begin {equation}\begin {split}  \label {Ot4}  &     \int_{ \{ x \in \Delta: d_{\partial \Delta}(x)  \le R \} } \rm {d }x \int_{  \Delta^c} \rm {d }y\frac { | v (x) - v (y)|^2} {|x-y|^{d+2s}}  \cr &=
 \int_{ \{ x \in \Delta: d_{\partial \Delta}(x)  \le R \} } \rm {d }x \left [  \int_{ \{y \in  \Delta^c,  d_{\partial \Delta}(y)  \le \rho\} } \rm {d }y\frac { | v (x) - v (y)|^2} {|x-y|^{d+2s}}   + 
     \int_{ \{y \in  \Delta^c,  d_{\partial \Delta}(y)  \ge\rho \}} \rm {d }y\frac { | v (x) - v (y)|^2} {|x-y|^{d+2s}} \right ] 
     \end{split}   \end {equation}   
  The first term in   \eqref {Ot4} is estimated as following.   Since $v \in H^s$,  $ \alpha >s$,  set $\alpha = s +\gamma$,  $\gamma>0$,  $| v (x) - v (y)| \le |x-y|^{s+ \gamma}$
\begin {equation}  \begin {split} &
 \int_{ \{ x \in \Delta: d_{\partial \Delta}(x)  \le R \} } \rm {d }x \int_{ \{y \in  \Delta^c,  d_{\partial \Delta}(y)  \le \rho\} } \rm {d }y\frac { | v (x) - v (y)|^2} {|x-y|^{d+2s}}
  \cr & \le 
   \int_{ \{ x \in \Delta: \rho \le d_{\partial \Delta}(x)  \le R \} } \rm {d }x 
   \int_{ \{y \in  \Delta^c,  d_{\partial \Delta}(y)  \le \rho\} } \rm {d }y
   \frac { | v (x) - v (y)|^2} {|x-y|^{d+2s}}  
+
\int_{ \{ x \in \Delta:    d_{\partial \Delta}(x)  \le \rho \} } \rm {d }x \int_{ \{y \in  \Delta^c,  d_{\partial \Delta}(y)  \le \rho\} } \rm {d }y\frac { | v (x) - v (y)|^2} {|x-y|^{d+2s}}\cr &
\le  C (\rho) L^{d-1} R  + 
\int_{ \{ x \in \Delta:    d_{\partial \Delta}(x)  \le \rho \} } \rm {d }x \int_{ \{y \in  \Delta^c,  d_{\partial \Delta}(y)  \le \rho\} } \rm {d }y\frac { C} {|x-y|^{d- \gamma}} \cr & \le
 C (\rho) L^{d-1} R + \rho^{1+\gamma} L^{d-1} C \simeq \e  |\La|
\end{split}   \end {equation} 
 
 The second term in  \eqref {Ot4} is  bounded as following
\begin {equation}\begin {split}  & \int_{ \{ x \in \Delta: d_{\partial \Delta}(x)  \le R \} } \rm {d }x \int_{ \{y \in  \Delta^c,  d_{\partial \Delta}(y)  \ge \rho \}} \rm {d }y\frac {C} {|x-y|^{d+2s}}  \cr & \le
    R L^{d-1}    \int_{\{ |y| \ge \rho   \}} \rm {d }y\frac {C} {|y|^{d+2s}}   \le  R L^{d-1} C(\rho) \le \frac \e 2  |\Delta|
   \end{split}   \end {equation} 
  provided  $ L$  is suitable chosen.  
   \end {proof}

 \begin{prop} \label{convest} Take  $D \Subset \R^d $  and  assume that $u_n\to u$ in $C^{0,\alpha}(D)$ for $s<\alpha<2s,$ then
 $$
 I_n:=\left|\int_{D\times D}\frac{|u(z)-u(z')|^2-|u_n(z)-u_n(z')|^2}{|z-z'|^{d+2s}}dz dz'\right|\le C'(d)K |D|({\rm diam }(D))^d \|u-u_n\|_{C^{0,\alpha}}\to 0.
 $$
 \end {prop}
 \begin {proof}  
 Note that \begin{eqnarray*}&&\left||u(z)-u(z')|^2-|u_n(z)-u_n(z')|^2\right|=\left| [ u(z)-u(z')+u_n(z)-u_n(z')] [u(z)-u(z')-(u_n(z)-u_n(z'))]\right|
 \\
 &\le&\left(|u(z)-u(z')|+|u_n(z)-u_n(z')|\right)\left(\left|(u(z)-u_n(z))-(u(z')-u_n(z'))\right|\right)\\&\le& (\|u\|_{C^{0,\alpha}}+\|u_n\|_{C^{0,\alpha}})
 |z-z'|^\alpha\cdot \|u-u_n\|_{C^{0,\alpha}} |z-z'|^\alpha.
 \end{eqnarray*}
 As a convergent sequence is bounded, there is a $K>0$ such that $(\|u\|_{C^{0,\alpha}}+\|u_n\|_{C^{0,\alpha}})<K.$ So
 \begin{eqnarray*}
 I_n&\le&\int_{D\times D}\frac{\left||u(z)-u(z')|^2-|u_n(z)-u_n(z')|^2\right|}{|z-z'|^{d+2s}}dz dz'\le
K \|u-u_n\|_{C^{0,\alpha}}\int_{D\times D}|z-z'|^{2(\alpha-s)-d}\\ &&\le C(d)K |D| \|u-u_n\|_{C^{0,\alpha}}\int_0^2r^{\delta-1}dr\le
C'(d)K |D|({\rm diam }(D))^d \|u-u_n\|_{C^{0,\alpha}}\to 0
 \end{eqnarray*}where $\delta=2(\alpha-s)>0$. 
 Note that we need only $\alpha>s$.

  \end {proof}

 \section {Appendix}  

We collect in this section  general results about  fractional laplacian scattered in the literature
and  recall the  main probabilistic  result  used to prove  Theorem  \ref {A3}. 
 
 \subsection  { Minimizers of   the functional  \eqref {funct110}  on open  bounded   Lipschitz  sets.}

 We recall here some basic results    assuring  that the minimization of  the functional  \eqref  {funct110}   in an  open, bounded   Lipschitz  set  has solution.    In the following $\om$ plays the role of a parameter.  It is kept fixed
   and the results hold for  all $\om\in \Om$.
 
  \begin{prop} \label{jt1}   Let  $\Lambda \Subset \R^d$  be  a Lipschitz bounded open set  and    $u_0: \R^d \to 
  \R$ be a measurable function.  Suppose that there exists a measurable function $ \tilde u$ which coincides with $u_0$  in  $\La^c$  and such that $ G_1(\tilde u, \om, \La) < \infty$.   Then there exists a measurable function $u^*$ such that 
 $$ G_1^{u_0}(u^*, \om, \La)  \le G_1^{u_0}(v, \om, \La) $$ 
 for any measurable function $v$  which coincides with $u_0$ in   $\La^c$. 
 \end {prop}
 \begin {proof}  Take a minimizing sequence, that is, let $u_k= u_0$  in $\La^c$ so that
  $  G_1(u_k, \om, \La)  \le  G_1(\tilde u, \om, \La) $ and 
  $$ \lim_{k \to \infty}  G_1(u_k, \om, \La)=  \inf_v  G_1(v, \om, \La)$$
  for any  $v$  which coincides with $u_0$ in in  $\La^c$. 
 Then by the following compactness result, see Proposition \ref {jt1a},  up to subsequence, $u_k$ converges almost everywhere to some $u^*$. By  Fatou's Lemma we conclude.
 \end {proof}
 
  \begin{prop}
   \label{jt1a}    Let $\La \Subset \R^d$ be a Lipschitz open set and $ \FF$  be a bounded subset of $L^2 (\La)$. Suppose that
  $$\sup_{f \in \FF}  \int_{\La} \rm {d }x \int_{\La}  \rm {d }y\frac { | f (x) - f (y)|^2} {|x-y|^{d+2s}}   < \infty.$$
  Then $\FF$ is precompact in $L^2 (\La)$.
  \end {prop}
For the   proof of   Proposition  \ref  {jt1a}   see    \cite  [Lemma 6.11]  {PSV}.   The proof is based 
on the classical Riesz-Frechet-Kolmogorov Theorem.  Some modifications are needed due to the non-locality of the fractional norm.  If $\La$  is not  Lipschitz  then  Proposition \ref {jt1a} does not hold. One can find counterexample, see for example  \cite [Example 9.2] {DPV1}. 

   Next we show that minimizers of the functional  \eqref  {funct110}  solve  the Euler -Lagrange equation  \eqref {EL.1}
  and  prove some regularity results.

In the following, $\La$, $v_0$ and $\omega$ are kept fixed, therefore we 
write $  G_1^{v_0}(v,\omega,\Lambda )= G_1(v)$. 
To derive the   Euler Lagrange equation for  the minimizers of $G_1(v)$ we compute the Frechet derivative of $G_1 (v)$.
   For   
   $ w \in C^\infty_0 (\La)$ we have that 
  \begin{equation} \label {ct1} \begin {split}   &G_1 (v + t w) =  G_1 (v)  +2 t    \int_{\Lambda} \rm {d }x \int_{\Lambda} \rm {d }y\frac { [ v(x) - v(y)] \cdot  [ w(x) - w(y)] } {|x-y|^{d+2s}}   \cr &  +  t  \int_{\Lambda} \rm {d }x[ W' (v) + \theta g_1]  w + 
 4 t      \int_{\Lambda} \rm {d }x \int_{\Lambda^c} \rm {d }y\frac { [ v(x) - v_0(y)] \cdot   w(x)   } {|x-y|^{d+2s}}  + O(t^2), \end {split}  \end {equation} 
 where $ W'(\cdot )$ is the derivative of $W (\cdot)$  with respect to its argument. 
  Then   the Frechet derivative  computed in $v$      is  the following linear operator defined for $ w \in C^\infty_0 (\La)$
  as the following
  \begin{equation}  \label {t4}\begin {split} &
   D_vG_1   (w) =  
  2     \int_{\Lambda} \rm {d }x \int_{\Lambda} \rm {d }y\frac { [ v(x) - v(y)] \cdot  [ w(x) - w(y)] } {|x-y|^{d+2s}}   \cr &  +     \int_{\Lambda} \rm {d }x[ W'(v) +\theta g_1]  w + 
 4        \int_{\Lambda} \rm {d }x \int_{\Lambda^c} \rm {d }y\frac { [ v(x) - v_0(y)]     w(x)   } {|x-y|^{d+2s}}.  \end {split}  \end {equation} 
At this point one is  tempted to split the first  integral in \eqref {t4}  in two terms and exchange $x$ with $y$ in one of the terms  to obtain
 \begin{equation}  \label {t5}  2       \int_{\Lambda} \rm {d }x w(x) \int_{\Lambda} \rm {d }y\frac { [ v(x) - v(y)]   } {|x-y|^{d+2s}} .    \end {equation}
However we cannot always do that. The inner integral in  the first  integral in \eqref {t4} might not be absolutely  convergent. So in general it can be defined only as a principal value. In such a case
 \begin{equation}  \label {t6} \int_{\Lambda} \rm {d }x \int_{\Lambda} \rm {d }y\frac { [ v(x) - v(y)] \cdot  [ w(x) - w(y)] } {|x-y|^{d+2s}}   =  \int_{\Lambda} \rm {d }x w(x)  \lim_{r \to 0}\int_{\Lambda \setminus B_r(x)} \rm {d }y\frac { [ v(x) - v(y)]   } {|x-y|^{d+2s}},  \end {equation}
where $B_r(x)$ is a ball of radius $r>0$ centered in $x$. 

From \eqref {t4} and \eqref {t6}  we  deduce  that  a  minimizer  of $ G_1^{v_0}(v,\omega,\Lambda ) $   is  a  function $v \in    H^s_{loc} \cap L^\infty $  which solves 
  \begin{equation}  \label {t23} \begin {split} & 2       \int_{\Lambda} \rm {d }x w(x) ((-\Delta)^s v) (x) \cr &  +     \int_{\Lambda} \rm {d }x[ W'(v) +\theta g_1]  w + 
 4      \int_{\Lambda} \rm {d }x \int_{\Lambda^c} \rm {d }y\frac { [ v(x) - v_0(y)]     w(x)   } {|x-y|^{d+2s}}=0.
 \end {split}  \end {equation} 
We identify the  problem stated in \eqref {t23}  to   the following Dirichlet  boundary value problem for  the corresponding  Euler-Lagrange equation:
 \begin{equation}  \label{EL.1} \begin{split}    
&  ( -\Delta)^s v = - \frac 1 {2} [W'(v)+\theta   g_1]    \quad 
\text{in } \Lambda,  \qquad   \om \in \Omega  \\ &
        v = v_0\quad  \text {in }    \La^c.
 \end{split}
\end{equation}
   
   We  recall  the following regularity result  proven in       \cite [ Proposition 2.9] {LS}.
   \vskip0.5cm 
  \noindent
   \begin{prop}\label{Lip}   Let    $ w= (-\Delta)^s u $  in $\R^d$    so that $ \| u \|_\infty $ and $\|w\|_\infty$   are finite.   If $ 2s \le 1$  then $u \in C^{0,\alpha}  $ for any $\alpha < 2 s$, 
and
$$ \|u\|_{C^{0,\alpha}} \le C [ \|u\|_\infty +      \|w\|_\infty ]  $$ 
for a constant $C=C(d,s, \alpha)$.

If $ 2s >1$  then $u \in C^{1,\alpha}   $ for any $ \alpha < 2s-1$, 
and
$$ \|u\|_{C^{1,\alpha}} \le C [ \|u\|_\infty +      \|w\|_\infty ], $$ 
for a constant $C=C(d,s, \alpha)$.
\end{prop}  
  We remark that the above results are valid for solution of \eqref {EL.1} in bounded domains, leading to a local regularity theory. 
 
\vskip0.5cm

The  main tool to prove  Theorem   \ref {A3} is the  following general result
which we reported from  \cite {HH}, see        \cite[Theorem 3.2 and Corollary 3.1] {HH}.   
  \begin {thm}   Let $S_{n,i}$,  $i=1, \dots k_n$ be  a double array of   zero  mean martingales  with respect to the filtration  $\FF_{n,i}$, $ \FF_{n,i} \subset \FF_{n+1,i}$ $i=1, \dots k_n$ with $S_{n,k_n}= S_n$,   so that $S_{n,i}= \E[ S_n| \FF_{n,i}] $.   We assume  that $k_n \uparrow \infty$  as $n \uparrow \infty$. 
Denote 
$$ X_{n,i}:= S_{n,i}- S_{n,i-1}, $$
$$ V_{n}= \sum_{i=1}^{k_n} \E [  X^2_{n,i}| \FF_{n,i-1}],   $$
 $$ U_{n,a}= \sum_{i=1}^{k_n} \E [  X^2_{n,i} \1_{\{ |[  X^2_{n,i}| >a\}} |  \FF_{n,i-1}]. $$
 Suppose that
 \begin {itemize}
 \item  for some constant $b^2$ and for  all  $\delta>0$,  $\lim_{n \to \infty} \Pr [ |V_n-b^2| \ge \delta] =0 $, 
 \item  for any $a>0$ and for any $ \delta >0$ 
   $$ \lim_{n \to \infty} \Pr \left [ U_n(a) \ge \delta \right] =0, \qquad\qquad   (Lindeberg\ condition)$$
 \end {itemize}
 then in distribution 
  $$ \lim_{n \to \infty}  S_n \stackrel {D} {=}  Z,$$
  where $Z$ is a   Gaussian random 
  variable with mean equal to zero and variance equal to $b^2$.  
\end {thm}


\begin{thebibliography}{19}




\bibitem{AW} M. Aizenman and J. Wehr, {\em Rounding effects 
on quenched randomness on first-order phase transitions}, 
Comm. Math. Phys., {\bf 130} (1990), pp. 489--528.










\bibitem{B} A. Bovier, {\em Statistical Mechanics of Disordered Systems.
A Mathematical Perspective}, Chemical Physics   {\bf 284}  (2002), pp.  409-421


\bibitem{BS} D. Brockmann and I. M.   Sokolov, {\em L\'evy flights in external force fields: from models to equations.
A Mathematical Perspective}, Cambridge University Press,  2006.


\bibitem{DPV1}  E. Di Nezza, G. Palatucci, E. Valdinoci {\em Hitchhiker's guide to the fractional Sobolev spaces}   Bull. Sci. Math. 136 (2012), no. 5, pp. 521-573.


\bibitem{D} P.D. Ditlevsen {\em Anomalous  jumping in a double-well potential},   Phy. Rev. E  {\bf 60}  No1 (1999), pp.  172-179

 






\bibitem{DO} N. Dirr and E. Orlandi     {\em 
Sharp-interface limit of a Ginzburg-Landau functional
 with a random external field},      
SIAM J. Math. Anal. {\bf 41}, Issue 2, (2009), pp. 781-824.


 \bibitem{DO2} N. Dirr and E. Orlandi     {\em 
Unique minimizer for a random functional with double well potential in dimension 1 and 2},      
Commun. Math. Sci {\bf  1}, Issue 9, (2011),  pp. 331-351


 \bibitem{GM} A. Garroni  and S. M\"uller     {\em 
A variational model for dislocation in the line tension limit },      
Arch. Rational Mech. Anal.,  {\bf 181}, No 3, (2006), pp. 535--578.

 \bibitem{GP} A. Garroni  and G. Palatucci     {\em 
A  singular perturbation result with a fractional norm},      
 in Variational Problem in Material Sciences, 
 Progr. Nonlinear Differential equation and Their Applications {\bf 68}, (2006), pp 111-126







 \bibitem{Go}  M.d.M Gonzales    {\em 
Gamma convergence of  an energy  functional related to the fractional Laplacian},      
Calc. Var. Partial Diff. Equ.,  {\bf 36}, no 2,   (2009), pp. 173--210.



\bibitem{HH}
P.  Hall and C. C. Heyde,
{\em Martingale limit theory and its application},
New York. Academic Press,   (1980).


\bibitem{GK} G. Keller, {\em  Equilibrium States in Ergodic Theory }
Londom Math. Society, Student Texts 42 (1998)



\bibitem{MK} R. Metzler and J. Klafter, {\em  The random walk's guide to anomalous diffusion: a fractional dynamical approach. }
Phys. Rep-Rev. Sect. Phys.  Lett.  {\bf 339},  No 1,   (2000) pp. 1-77.



\bibitem{PSV} G. Palatucci, O. Savin and  E. Valdinoci { Local and global minimizers for a variational energy involving fractional norm}   Ann. Mat. Pura  Appl. {\bf 92} (2013), no 4, pp 673-718 

\bibitem{RP}
Rama Cont and P. Tankov
{\em Financial modelling  with jump processes.},
 Chapman \& Hall, CRC Financial Mathematics Series. Chapman \&Hall, CRC, Boca Ratou, FL,2004
 
 \bibitem{SV}  O. Savin and E. Valdinoci ,  {\em   $\G-$ convergence for nonlocal phase transitions}   
  Ann. Inst. H. PoincarŽ Anal. Non Lin\'eaire    {\bf 29}  (2012), no. 4, pp. 479-500.

  
 \bibitem{LS} Luis  Silvestre ,  {\em  Regularity of the obstacle problem for a fractional power of the  Laplacian  Operator,}    Comm.  Pure Appl. Math. {\bf 60} (2007) pp. 67-112
 

         



\bibitem{WGMN} B.J.West,  P. Grigolini, R. Metzler and T. Nonnenmacher  {\em  Fractional Diffusion and  L\'evy stable processes,}    Phys. RevE. {\bf 55} No1(1997) pp. 99-106


\end{thebibliography}
\end{document}